\documentclass[letterpaper,utf8,11pt]{article}
\usepackage[utf8]{inputenc}
\usepackage{algpseudocode}
\usepackage[margin=1in]{geometry}
\usepackage{algorithm}
\usepackage{authblk}
\usepackage{geometry}
\usepackage{amsthm,amsmath,amssymb}
\usepackage{xcolor}
\usepackage{bbm}
\usepackage{xfrac}
\usepackage{thm-restate}
\usepackage{xspace}
\usepackage{palatino}
\usepackage{enumitem}
\usepackage{bm}
\usepackage{multicol}
\newcommand{\old}[1]{{}}
\usepackage{hyperref}
\hypersetup{pdftex, plainpages = false, pdfpagelabels, 
bookmarks=false,
bookmarksopen = true,
bookmarksnumbered = true,
breaklinks = true,
linktocpage,
pagebackref,
colorlinks = true,  % was true
linkcolor = blue,
urlcolor  = blue,
citecolor = purple,
anchorcolor = green,
hyperindex = true,
hyperfigures
}

\usepackage{cleveref}
\usepackage{mathtools}
\usepackage{multirow}
\usepackage{colortbl}
\usepackage{hhline}
\usepackage[textwidth=1.1in]{todonotes}
\makeatletter
\newcommand\tinyv{\@setfontsize\tinyv{4pt}{6}}
\makeatother
\usepackage{tcolorbox}
\usepackage{makecell}
\theoremstyle{plain}
\newtheorem{theorem}{Theorem}
\newtheorem{lemma}[theorem]{Lemma}

\newtheorem{observation}[theorem]{Observation}
\newtheorem{claim}{Claim}[theorem]

\theoremstyle{definition}
\newtheorem{definition}[theorem]{Definition}

\newtheorem{question}{Question}

\newcommand{\agi}[1]{\todo[inline,size=\tiny]{ AG: #1}}
\newcommand{\lr}[1]{\left(#1\right)}
\newcommand{\LR}[1]{\left\{#1\right\}}

\newcommand{\cP}{\mathcal{P}}
\newcommand{\cO}{\mathcal{O}}

\newcommand{\cC}{\mathcal{C}}
\newcommand{\cT}{\mathcal{T}}
\newcommand{\cost}{{\sf cost}\xspace}
\newcommand{\wcost}{{\sf wcost}\xspace}
\renewcommand{\epsilon}{\varepsilon}
\newcommand{\opt}{{\sf OPT}}

\renewcommand{\d}{{\sf d}}
\renewcommand{\c}{6}
\newcommand{\findassgn}{\texttt{ASSIGNMENT}}
\newcommand{\ballint}{\texttt{BALL-INT}}
\newcommand{\genballint}{\texttt{GenBALL-INT}}
\newcommand{\epasalgo}{\texttt{COMPUTE-SOLUTION}}
\newcommand{\cI}{\mathcal{I}}
\newcommand{\cM}{\mathcal{M}}

\newcommand{\cG}{\mathcal{G}}
\newcommand{\cB}{\mathcal{B}}
\newcommand{\coresetsize}{\ell}
\newcommand{\assgntime}{T_\textsf{A}}
\newcommand{\ballinttime}{T_\textsf{B}}

\newcommand{\sclen}{O(\lambda(\tfrac{\epsilon}{2}) \tfrac{\log (1/\epsilon)}{\epsilon} ) }
\newcommand{\pathlen}{O(k\lambda(\tfrac{\epsilon}{2}) \tfrac{\log (1/\epsilon)}{\epsilon} ) }
\newcommand{\algotime}{2^{O(\tfrac{k}{\epsilon}\log (k\coresetsize) \cdot \lambda(\tfrac{\epsilon}{2}) \log(\sfrac{1}{\epsilon}))}\poly(|\cI|)\cdot \assgntime(\cI)\cdot\ballinttime(\cI)}
\newcommand{\epastime}{2^{O\left(\tfrac{k}{\epsilon}\log (k\coresetsize) \cdot \lambda(\tfrac{\epsilon}{60z}) \log(\sfrac{1}{\epsilon})\right)}\poly(|\cI|)\cdot \assgntime(\cI)\cdot\ballinttime(\cI)}

\newcommand{\ball}{{\sf ball}}
\newcommand{\probname}{$(k,z)$-Clustering}
\DeclareMathOperator{\poly}{poly}

\definecolor{MyBlue}{rgb}{0.200,0.200,0.600}
\newcommand{\consprob}{Well Constrained $(k,z)$-Clustering}
\newcommand{\conskmed}{Well Constrained $(k,z)$-Median}
\title{\vspace{-2cm}Clustering under Constraints:\\ Efficient Parameterized Approximation Schemes}

\vspace{1cm}
  \author{Sujoy Bhore\thanks{Department of Computer Science \& Engineering, Indian Institute of Technology Bombay, Mumbai, India.\\ Email: \href{sujoy@cse.iitb.ac.in}{sujoy@cse.iitb.ac.in}}
  \qquad Ameet Gadekar\thanks{CISPA Helmholtz Center for Information Security, Saarbr\"{u}cken, Germany.\\ Email: \href{ameet.gadekar@cispa.de}{ameet.gadekar@cispa.de}} \qquad Tanmay Inamdar\thanks{Department of Computer Science \& Engineering, Indian Institute of Technology Jodhpur, Jodhpur, India.\\ Email: \href{taninamdar@gmail.com}{taninamdar@gmail.com}}	
  }
\date{}

\begin{document}

\maketitle
%\vspace{-1cm}

%\todo[size=\tiny]{Title suggestion: ``Coreset Strikes Back: Parameterized Approximation Schemes for (Constrained) $k$-Clustering Made Faster and Simpler.}

\begin{abstract}
Constrained $k$-clustering is a general framework for clustering that captures a wide range of structural and assignment constraints, including capacities, fairness, outliers, and matroid restrictions.
Despite its broad applicability, efficient parameterized approximation schemes (EPASes) for these problems were previously known only in continuous Euclidean spaces or for formulations restricted to Voronoi assignments, where each point is assigned to its nearest center. Moreover, it is known that EPASes cannot exist in general metric spaces.

We present a unified framework that yields EPASes for constrained $(k,z)$-clustering in metric spaces of bounded (algorithmic) scatter dimension, a notion introduced by Abbasi et al.~(FOCS 2023). They showed that several well known metric families, including continuous Euclidean spaces, bounded doubling spaces, planar metrics, and bounded treewidth metrics, have bounded scatter dimension. Subsequently, Bourneuf and Pilipczuk~(SODA 2025) proved that this also holds for metrics induced by graphs from any fixed proper minor closed class. Our result, in particular, addresses a major open question of Abbasi et al., whose approach to $k$-clustering in such metrics was inherently limited to \emph{Voronoi-based} objectives, where each point is connected only to its nearest chosen center.

As a consequence of our framework, we obtain EPASes for several constrained clustering problems, including capacitated and matroid $(k,z)$-clustering, fault tolerant and fair $(k,z)$-clustering, as well as for metrics of bounded highway dimension. In particular, our results on capacitated and fair $k$-Median and $k$-Means provide the first EPASes for these problems across broad families of structured metrics. Previously, such results were known only in continuous Euclidean spaces, due to the works of Cohen-Addad and Li (ICALP 2019) and Bandyapadhyay, Fomin, and Simonov (ICALP 2021; JCSS 2024), respectively. Along the way, we also obtain faster EPASes for uncapacitated $k$-Median and $k$-Means, improving upon the running time of the algorithm by Abbasi et al. (FOCS 2023).

%We present a unified framework that yields EPASes for constrained $(k,z)$-clustering in metric spaces of bounded algorithmic scatter dimension, a notion of metric spaces that unifies several well-known families, including continuous Euclidean space, bounded doubling space, planar and bounded-treewidth metrics~(Abbasi et al., FOCS~2023), as well as metrics induced by graphs from any fixed proper minor-closed graph class~(Bourneuf and Pilipczuk, SODA~2025). This, in particular, addresses a major open question of Abbasi et al.~(FOCS~2023), who obtained EPASes for $k$-clustering in such metrics, but whose approach was inherently limited to objectives based on \emph{Voronoi assignments}, where each point is connected only to its nearest chosen center. 

%As a consequences of our framework, we obtain EPASes for Capacitated and Matroid $(k, z)$-Clustering, EPASes for bounded highway dimension, EPAS for Fault Tolerant $(k, z)$-Clustering, and Fair $(k, z)$-Clustering. In particular, our results on \emph{capacitated} and \emph{fair} $k$-Median and $k$-Means make substantial progress, providing the first EPASes for these problems across broad families of structured metrics. Previously, such results were known only in continuous Euclidean spaces, due to the works of Cohen-Addad and Li~(ICALP~2019) and Bandyapadhyay, Fomin, and Simonov~(ICALP~2021; JCSS~2024), respectively. Along the way, we obtain faster EPASes for \emph{uncapacitated} $k$-Median/Means, improving upon the running time of the algorithm by Abbasi et al. (FOCS 2023). 

\old{

Algorithmic scatter dimension is a notion of metric spaces, introduced recently by Abbasi et al. (FOCS 2023), which unifies various well-known metric spaces, including continuous Euclidean space, bounded doubling space, planar and bounded treewidth metrics. Recently, Bourneuf and
Pilipczuk (SODA 2025) showed that metrics induced by graphs from any fixed proper minor closed graph class have bounded scatter dimension. 
Abbasi et al. presented a unified approach to obtain EPASes (i.e., $(1+\epsilon)$-approximations running in time FPT in $k$ and $\epsilon$) for $k$-Clustering in metrics of bounded scatter dimension. However, a seemingly inherent limitation of their approach was that it could only handle clustering objectives where each point was assigned to the closest chosen center, i.e., \emph{Voronoi clustering}. They explicitly asked, \emph{if there exist EPASes for constrained $k$-Clustering for metrics of bounded scatter dimension}. 

We present a unified framework which yields EPASes for metrics of bounded algorithmic scatter dimension, for \emph{capacitated} and \emph{fair} $k$-Median/Means. Our framework exploits coresets for such constrained clustering problems in a novel manner, 
and notably, requires only coresets of size $\left(k \log n\epsilon^{-1}\right)^{O(1)}$, which are usually constructible even in \emph{general metrics}.
Furthermore, we note that due to existing lower bounds it is impossible to obtain such an EPAS for Capacitated $k$-Center, thereby, essentially answering the complete spectrum of the question.  

Our results on \emph{capacitated} and \emph{fair} $k$-Median/Means make significant progress, and provide the first EPASes for these problems in  broad families of metric spaces. Earlier such results were only known in continuous Euclidean spaces due to the work of Cohen-Addad \& Li, (ICALP 2019), and Bandyapadhyay, Fomin \& Simonov, (ICALP 2021; JCSS 2024), respectively. Along the way, we obtain faster EPASes for \emph{uncapacitated} $k$-Median/Means, improving upon the running time of the algorithm by Abbasi et al. (FOCS 2023). 
}

%On the other hand, a seemingly inherent limitation of the approach of Abbasi et al. was that they could only handle clustering objectives involving norms, where each point was assigned to the closest chosen center, i.e., voronoi clustering. Our results are the first general EPASes for a broad spectram of metric spaces that goes beyond non-voronoi assignments. 

\end{abstract}

%\begin{itemize}
 %   \item Capacitated results - First EPASes for $k$-Means/Median for scatter dimension - includes... Earlier it was only known for Euclidean. (these are deterministic modulo coreset construction)   
  %  \item Faster EPASes for uncapacited $k$-Means/Median. Improves upon the FOCS'23 paper running time. (this is randomized) 
  %  \item Mention that not possible for bounded highway dimension ... so this is the limit that it could have been pushed using our technique. 
   % \item Applications of our framework. 
%\end{itemize}

% \newpage
%  \tableofcontents
{
	\vfill
	\begin{multicols}{2}
		{%\small 
			\setcounter{secnumdepth}{5}
			\setcounter{tocdepth}{2} \tableofcontents
		}
	\end{multicols}
  }

%\newpage

% !TeX root = main.tex
\section{Introduction}

Clustering and facility location problems form a central theme in modern algorithmic theory. Given a set of points with pairwise distances, the goal is to choose a small number of representative centers and assign points to them in a way that minimizes an aggregate measure of distance. This paradigm captures a wide range of applications in data analysis, machine learning, and operations research. Traditionally, clustering is viewed as partitioning the dataset, where each point belongs to a single cluster represented by its nearest center. A broader and more flexible perspective, however, treats these problems as assignment optimization: each point can distribute its weight among multiple facilities under certain structural or resource constraints (see, e.g., ~\cite{ShmoysTA97,CharikarG05, KleinbergT02,BalcanBG09}).

Among the best known formulations of clustering are the classical \emph{$k$-Median} and \emph{$k$-Means} problems, which have been studied intensively across theoretical computer science, machine learning, and optimization for several decades (see, e.g., ~\cite{kanungo2002efficient, DBLP:conf/stoc/Har-PeledM04, badoiu-etal:approximate-clustering-coresets, DBLP:conf/stoc/Cohen-AddadEMN22, KumarSS10}). These models seek $k$ centers that minimize the total (or squared) distance from each point to its assigned center. Over time, several more realistic and constrained variants have been explored. Among these, the \emph{Capacitated} $k$-\emph{Median} and $k$-\emph{Means} problems stand out, where each center can serve only a limited number of points or must satisfy additional combinatorial restrictions on the set of open centers. Such constrained variants naturally arise in network design, resource allocation, and facility planning, and are substantially harder from an algorithmic standpoint. Despite decades of work, Capacitated $k$-Median remains one of the central open problems in approximation algorithms: the best known approximation ratio in general metrics is $O(\log k)$, obtained via probabilistic embeddings of metrics into trees~\cite{AdamczykBMM019}. Constant factor approximations are known only when either the capacities~\cite{Li17, Li16} or the number of clusters~\cite{DemirciL16, ByrkaRU16} are allowed to be violated by a $(1+\varepsilon)$ factor. On the hardness side, the strongest lower bound known, $(1+\frac{2}{e}-\varepsilon)$ under Gap-ETH, is inherited directly from the uncapacitated version. These gaps between upper and lower bounds have motivated the search for improved approximations in structured metrics, often by allowing fixed parameter or quasi polynomial running time~\cite{Cohen-AddadL19Capacitated, BhattacharyaJK18}.

In recent years, the connection between the geometric structure of metrics and the design of parameterized approximation algorithms has become increasingly prominent. Cohen-Addad and Li~\cite{Cohen-AddadL19Capacitated} obtained a $(3+\varepsilon)$ approximation for Capacitated $k$-Median in general metrics with running time fixed parameter tractable (FPT) in $k$ and $\varepsilon$. Notably, for the uncapacitated case, the lower bound of $(1+\frac{2}{e})$ can be matched by an FPT algorithm~\cite{Cohen-AddadG0LL19Tight}. A central question thus arises: \emph{for which families of metric spaces can one obtain better upper bounds for Capacitated $k$-Median or $k$ Means?} The only affirmative progress is known for Euclidean metrics, where Cohen-Addad and Li showed that $(1+\varepsilon)$ approximation algorithms exist for both problems with running time $f(\varepsilon,k)\,n^{O(1)}$~\cite{Cohen-AddadL19Capacitated} (see also~\cite{BhattacharyaJK18}).\footnote{In this context, a recent work of Huang et al.~\cite{Huang0L025} on coresets for constrained clustering is also relevant; we discuss its implications in \Cref{subsec:remarks}.}

As mentioned before, stronger upper bounds are known for the uncapacitated versions of the problem. Indeed, Cohen-Addad et al.~\cite{Cohen-AddadG0LL19Tight} obtained tight $(1+\frac{2}{e}+\varepsilon)$ and $(1+\frac{8}{e}+\varepsilon)$ approximation factors for $k$-Median and $k$-Means, respectively, in general metric spaces, with running times fixed parameter tractable (FPT) in $k$ and $\varepsilon$. However, an important question remains: \emph{can these lower bound barriers for uncapacitated $k$-Median and $k$-Means be broken for structured metric spaces?} This question has driven a series of advances in Euclidean settings, where geometric structure enables substantially stronger approximation guarantees. B\u{a}doiu, Har-Peled, and Indyk~\cite{badoiu-etal:approximate-clustering-coresets} showed that, for points in $\mathbb{R}^d$, a $(1+\varepsilon)$-approximate $k$-Median can be computed in expected time $2^{(k/\varepsilon)^{O(1)}}d^{O(1)}n \log^{O(k)} n$. Har-Peled and Mazumdar~\cite{Har-PeledM05} subsequently proved the existence of small coresets for both $k$-Median and $k$-Means in low dimensions, while Kumar, Sabharwal, and Sen~\cite{KumarSS10} gave simple randomized algorithms achieving $(1+\varepsilon)$ approximations with constant probability in $O(2^{(k/\varepsilon)^{O(1)}}dn)$ time. A common theme among these works is the use of geometric sampling and coreset constructions that exploit Euclidean structure. Most of these algorithms address the continuous setting, where centers may be placed anywhere in $\mathbb{R}^d$, although the result of~\cite{KumarSS10} also applies to the discrete $k$-Means problem.

Over the past decade, coreset techniques have evolved into a general algorithmic framework for clustering across diverse metric spaces. Feldman et al.~\cite{feldman2013turning} constructed coresets for $k$-Means, PCA, and projective clustering whose sizes are independent of the dimension, and Sohler and Woodruff~\cite{sohler2018strong} extended these ideas to $k$-Median and subspace approximation. Braverman et al.~\cite{braverman2021coresets} developed coresets for $k$-Median in excluded-minor metrics. 
%extending their applicability to Euclidean spaces via the terminal embedding theorem of Narayanan and Nelson~\cite{NarayananN19}, which generalizes the Johnson–Lindenstrauss Lemma. 
After a long sequence of advances~\cite{feldman2011unified, DBLP:conf/icml/BakerBHJK020, DBLP:conf/stoc/HuangV20, braverman2021coresets}, Cohen-Addad et al.~\cite{cohen-addad-etal21:coreset-framework} presented a unified coreset framework encompassing Euclidean, doubling, and minor-free metrics, as well as general metric spaces. While these results provide remarkably small coresets and broad applicability, an important algorithmic challenge remains: 
\emph{Can coreset-based methods also yield unified framework for designing faster (EPASes) approximation schemes for (constrained) Clustering?}

%\todo[inline]{TI to Sujoy: There was a reviewer comment about defining $k$-clustering. However, the last sentence is not clear (e.g., we are not first to investigate this question). Should we keep this line?}

\paragraph*{EPASes for (Constrained) Clustering.} 
An \emph{efficient parameterized approximation scheme} (EPAS)\footnote{A 
$(1+\varepsilon)$-approximation algorithm that runs in time $f(k,\varepsilon)\poly(n)$ for every $\varepsilon>0$.} is perhaps the most powerful synthesis of approximation and parameterized algorithms, and has become a central theme in recent algorithmic research (see~\cite{AbbasiClustering23, bandyapadhyay2024parameterized, matouvsek2000approximate, ostrovsky2013effectiveness, ding2020unified} and the references therein). However, for $k$-Median and $k$-Means, it is impossible to obtain EPASes for arbitrary metric spaces~\cite{Cohen-AddadG0LL19} under standard complexity-theoretic conjectures. Consequently, much of the literature has focused on designing EPASes for \emph{special metrics}, such as continuous Euclidean spaces~\cite{KumarSS10, badoiu-etal:approximate-clustering-coresets}. 

In a recent work, Abbasi et al.~\cite{AbbasiClustering23} presented a unified EPAS framework that applies to a variety of clustering objectives as well as diverse metric spaces. To achieve this, they introduced the notion of \emph{$\varepsilon$-scatter dimension} (see Section~\ref{sec:scatter-dim} for the formal definition), which encompasses many commonly studied metric families, including continuous Euclidean and doubling spaces. Subsequently, Bourneuf and Pilipczuk~\cite{BourneufP25} showed that metrics induced by graphs from any fixed proper minor-closed graph class also have bounded scatter dimension. Moreover, they developed metric analogs of well-known graph invariants from the theory of \emph{sparsity}, and as a consequence constructed a \emph{coreset} for $k$-Center in any proper minor-closed graph class whose size is polynomial in $k$, with the exponent depending on the graph class and $\frac{1}{\varepsilon}$. 

Abbasi et al.~\cite{AbbasiClustering23} bypassed the use of coresets to obtain their clustering results. They argued that deriving clustering algorithms via coresets faces information-theoretic limitations, particularly for $k$-Center clustering, since small coresets do not exist in high-dimensional Euclidean spaces~\cite{braverman2019coresets}. However, an inherent limitation of their approach is that it only applies to clustering objectives in which each point is assigned to its nearest chosen center, i.e., \emph{Voronoi clustering}.\footnote{A recent work of Gadekar and Inamdar~\cite{DBLP:conf/stacs/Gadekar025} also used scatter dimension to design an EPAS for a common generalization of $k$-Center and $k$-Median that is not captured by the general norm framework of~\cite{AbbasiClustering23}, further demonstrating the broader applicability of the scatter-dimension concept. However, the clustering objective considered in~\cite{DBLP:conf/stacs/Gadekar025} still only supports \emph{Voronoi} assignments.} 
Indeed, Abbasi et al. explicitly posed the following broad open question:
\begin{question}\label{ques}
\emph{Is it possible to design EPASes for constrained $(k,z)$-Clustering\footnote{Informally, the objective of the $(k, z)$-clustering includes the sum of $z$-th power of distances, where $1 \le z < \infty$ is fixed. For example, $z = 1, 2$ corresponds to $k$-Median and $k$-Means, respectively. A formal definition follows.} in metric spaces of bounded (algorithmic) scatter dimension, for assignment constraints such as capacities, fairness, and diversity?}
\end{question}

In fact, \cite{AbbasiClustering23} asked this question for an even more general constrained clustering problem; however that problem is so general that it also captures the \emph{Capacitated} $k$-\emph{Center} problem, which is known to be \textsf{W[1]}-hard to approximate within any factor arbitrarily close to $1$ in metric spaces of bounded highway dimension~\cite{FeldmannV22}. Furthermore, it has been observed that metric spaces of bounded highway dimension also have bounded scatter dimension; \footnote{See the talk on~\cite{AbbasiClustering23} at FOCS 2023: \url{https://focs.computer.org/2023/schedule/}.} implying that the original question of~\cite{AbbasiClustering23} cannot be answered in the affirmative. Therefore, we reinterpret their question by restricting our focus to constrained $(k,z)$-Clustering, while still avoiding the hardness barrier of~\cite{FeldmannV22}.

The above discussion motivates a more general view of clustering under constraints.
Rather than focusing on any specific objective or restriction, we consider the broad setting of \emph{constrained $(k,z)$-clustering}, which captures a wide spectrum of problems including capacitated, fault-tolerant, fair, outlier, and matroid-constrained variants.  Our goal is to develop a unified algorithmic framework that yields efficient parameterized approximation schemes for this family of problems in metric spaces of bounded (algorithmic) scatter dimension, thereby also addressing the underlying challenge posed in Question~\ref{ques}. Our formulation is deliberately generic: it views clustering as an \emph{assignment optimization} problem, where each data point may fractionally distribute its weight among multiple centers under feasibility constraints. This unified perspective underlies our results and enables reasoning about diverse constrained objectives within a single framework. 

We now formally define the problem.

%\textcolor{blue}{Mention that we can define more general constraints. Our goal is to do more general framework and we also answer Q~1.}

\subsection{Problem Formulation}
We consider a generalized notion of \probname, which we call \emph{constrained \probname}. 
The input consists of a set of points $P$ (also called clients) with $|P|=n$, a set of facilities $F$, and a metric space $M=(P\cup F,d)$. 
Each point $p\in P$ has a nonnegative weight $w(p)$, and we are given a positive integer $k$.  

The goal is to choose a center set $X\subseteq F$ with $|X|=k$, together with an assignment function 
$f: P\times X \rightarrow \mathbb{R}_{\ge 0}$ satisfying
\[
\sum_{x\in X} f(p,x)=w(p) \qquad \text{for all } p\in P,
\]
so as to minimize the clustering cost
\[
\cost(X)=\sum_{p\in P}\sum_{x\in X} f(p,x)\cdot d(p,x)^z,
\]
for a fixed constant $z\ge 1$.  
In addition, we impose two types of constraints:  
\begin{description}[leftmargin=*]
    \setlength\itemsep{-0.5em}
    \item[(i) \emph{center constraints}]that restrict the allowable center sets $X$, and
    \item[(ii) \emph{assignment constraints}]that restrict feasible assignments $f$ for a given feasible center set.
\end{description} 
This general framework captures a wide range of well-studied constrained clustering problems. Below, we describe several important special cases.%\todo{AG: made the definition formal!}

\begin{itemize}
    \item \textbf{Capacitated Clustering~\cite{Li17,DemirciL16,KhullerS00,CyganHK12,Cohen-AddadL19Capacitated}} In this setting, each facility $c \in F$ has a capacity $\eta_c$ such that the total points assigned to $c$ are at most $\eta_c$. This is captured by the constraint for all $c\in C$ that   $\sum_{p \in P} f(p,x) \le \eta_c$, for all $p \in P$.\footnote{We consider a general setting of capacitated clustering, similar to~\cite{Cohen-AddadL19Capacitated}, where points can be assigned fractionally to multiple facilities.}
    \item \textbf{Fault Tolerant Clustering~\cite{DBLP:journals/talg/HajiaghayiHLLS16,DBLP:journals/tcs/KhullerPS00,10.1007/978-3-642-13036-6_19}:} In this problem, each point $p \in P$ is assigned to $\ell$ closest centers in $X$ to $p$. This can be obtained by the constraint $f(p,x) \le \tfrac{w(p)}{\ell}$ for all $p \in P$ and $x \in X$, and considering the normalized cost $\sum_{p \in P} \sum_{x \in X} f(p,x) \cdot (d(p,x))^z$, which is  $(\tfrac{1}{\ell})^{th}$ factor of the original cost of $X$.
    \item \textbf{$(\alpha,\beta)$-Fair Clustering~\cite{CKLV17,BandyapadhyayFS24}:} This problem partitions $P$ into $\ell$ colors $P_1,\dots, P_\ell$, and we are given a range $[\alpha_i,\beta_i]$ for each color $i \in [\ell]$. The goal is to find a set $X$ of $k$ centers and a clustering  of $P$ such that, in each cluster, the proportion of color-$i$ points to the cluster points is in the range $[\alpha_i,\beta_i]$. This problem is captured by adding the following constraint on $f$: for each center $x \in X$ and each color $i \in [\ell]$, we have $\frac{\sum_{p \in  P_i} f(p,x) }{\sum_{p \in P} f(p,x) } \in [\alpha_i,\beta_i]$. 
    \item \textbf{Outlier Constraint~\cite{CharikarKMN01,DBLP:conf/soda/Chen08}: } Here, we are given additional parameter $M \le W$, where $W =\sum_{p} w(p)$, and the constraint is to assign all the weight of points except $M$  to $X$. This can be achieved by relaxing the constraint $\sum_{x \in X} f(p,x) = w(p)$ to $\sum_{x \in X} f(p,x) \le w(p)$, and adding the constraint $\sum_{p \in P}\sum_{x \in X} f(p,x) \ge W-M$.
    \item \textbf{Matroid constrained Clustering~\cite{KrishnaswamyLS18,DBLP:conf/soda/KrishnaswamyKNSS11,DBLP:journals/talg/Swamy16}:} In this problem, we are additionally given a matroid $\mathbb{M}$ over $F$, and the constraint is to have a center set $X$ that is an independence set in $\mathbb{M}$. A particular example of this problem is a fair variant of $k$-median~\cite{10.1145/TGOO}, where each facility has one of the $k$ colors, and there is a bound $\alpha_i$ for each color. The constraint is to choose $X$ such that it has exactly $\alpha_i$ centers of color $i$.
\end{itemize}

\old{\newpage
\section{Old Introduction}
Clustering plays an indispensable role in our understanding of very large scale data. In $k$-clustering problems, we are given a set $\mathcal{P}$ of data points in a metric space $\mathcal{M}$, and the goal is to partition $\mathcal{P}$ into $k$ partitions (namely, \emph{clusters}) such that nearby points should ideally be in the same part, and each cluster is represented by a center. 
Clustering problems have been studied extensively in theoretical computer science, machine learning, computer vision, among other fields. Among the various notions of clustering, \emph{$k$-Median} and \emph{$k$-Means} are arguably the most popular and have been studied algorithmically for decades~\cite{kanungo2002efficient, DBLP:conf/stoc/Har-PeledM04, BCPSS24, DBLP:conf/stoc/Cohen-AddadEMN22, badoiu-etal:approximate-clustering-coresets, KumarSS10}. 
Clustering under capacity constraints is a natural and fundamental problem, and hence, quite conceivably, \emph{Capacitated $k$-Median} and \emph{Capacitated $k$-Means} have garnered lot of attention in recent years; see~\cite{BFRS15, ChuzhoyR05, Li17, Li16, DemirciL16, CharikarGTS02, VC20}. 

\emph{Capacitated $k$-Median} and \emph{Capacitated $k$-Means} are inherently harder than the vanilla (\emph{Uncapacitated}) $k$-Means or $k$-Median, since there is hard constraint on the number of clusters as well as the number of points that can be assigned to each cluster. In fact, till-date the best known approximation for capacitated $k$-Median remains $O(\log k)$ (folklore), %\todo{Folklore, but proof was given in \cite{AdamczykBMM019} for completeness}
using the probabilistic embedding of metric spaces into trees (Adamczyk et al. gave an explicit proof in~\cite{AdamczykBMM019} for completeness). Constant-factor approximation algorithms exist when either the capacities~\cite{Li17, Li16}, or the number of clusters~\cite{DemirciL16, ByrkaRU16}, are allowed to be exceeded by a $(1+\varepsilon)$ factor, for some constant $\varepsilon>0$.
From the lower bounds side, surprisingly, no better lower bound than $(1+\frac{2}{e}-\epsilon)$ is known, which is inherited from the uncapacited version. Given this evident gap, it is natural to ask whether the approximation landscape can be improved for metric spaces with special properties, possibly at the cost of higher running time. While it remains an open question whether a constant-factor approximation algorithm exists for general metrics that runs in polynomial time, Cohen-Addad and Li~\cite{Cohen-AddadL19Capacitated}, in 2019, showed a $(3+\epsilon)$-approximation for the problem, running in time FPT in $k$ and $\epsilon$.
However, even in the uncapacitated case and with FPT running time, the lower bound remains at $1+\frac{2}{e}$, assuming Gap-ETH.
% Although, the lower bound however remains at $1+\frac{2}{e}$, assuming Gap-ETH,  even for uncapacitated case and with FPT time. 
Notably, for uncapacitated case, this lower bound can be matched by an FPT algorithm~\cite{Cohen-AddadG0LL19Tight}. 
A central open question, thus,  is: \emph{For which families of metric spaces, is it possible to obtain better upper bounds for capacitated $k$-Median/Means.}%\todo[size=tiny]{T: there are back to back similar sounding q's for cap and uncap. maybe one of them we can phrase not as a q.} 
The only improvement in the affirmative sense is known for Euclidean spaces. For instance, for Euclidean spaces, Cohen-Addad and Li showed that there exists $(1+\varepsilon)$-approximation algorithms for both problems with running time $f(\varepsilon,k) n^{O(1)}$~\cite{Cohen-AddadL19Capacitated} (see also~\cite{BhattacharyaJK18}). \footnote{In this context, a recent work of \cite{Huang0L025} on coresets for constrained clustering is also relevant. We discuss the implications of this work and its relation with our work in \Cref{subsec:remarks}.}

As mentioned before, stronger upper bounds are known for \emph{uncapacitated $k$-Median/Means}. Indeed, Cohen-Addad et al.~\cite{Cohen-AddadG0LL19Tight} obtained tight $(1+\frac{2}{e}+\varepsilon)$ and $(1+\frac{8}{e}+\varepsilon)$ approximation factors for $k$-Median and $k$-Means in general metric spaces, respectively, with running times FPT in $k$ and $\epsilon$. %These bounds are tight for uncapacitated $k$-Means and $k$-Medians, assuming Gap-ETH.
However, an important question remains: \emph{Can the lower bound barriers for uncapacitated 
$k$-Median/Means be broken for special metric spaces?} To that end, improved FPT-approximation algorithms have been developed for Euclidean spaces; see~\cite{badoiu-etal:approximate-clustering-coresets,KumarSS10}. B\u{a}doiu, Har-Peled and Indyk~\cite{badoiu-etal:approximate-clustering-coresets} showed that, for points in $\mathbb{R}^d$, a $(1+\varepsilon)$-approximate $k$-Median can be computed in expected $2^{{k/\varepsilon}^{O(1)}}d^{O(1)}n \log^{O(k)} n$ time. 
%\emph{Coresets} eventually turned out to be leading pathway for obtaining better approximation bounds for clustering problems in special metric spaces. 
Subsequently, Har-Peled and Mazumdar in their major work~\cite{Har-PeledM05}, showed the existence of small coresets for the problems of computing $k$-Median and $k$-Means clustering for points in low dimension. 
Kumar, Sabharwal and Sen in their breakthrough work~\cite{KumarSS10} gave simple randomized algorithms for $k$-Median/Means that obtain $(1+\varepsilon)$-approximation solutions with constant probability, running in time $O(2^{{(k/\varepsilon)}^{O(1)}} dn)$. A fundamental aspect of these works is that they primarily rely on key sampling ideas (or, roughly speaking, \emph{coresets}) that crucially exploits properties of the Euclidean space. Moreover, all these  works primarily focus on the \emph{continuous} Euclidean cases (where centers can be opened anywhere), except that the result of Kumar et al.~\cite{KumarSS10} also works for discrete $k$-Means. 
%Discrete $k$-Means result of Kumar et al\todo{S: we should make a comment why it works for discrete $k$-means}\todo[size=small]{T: I don't know the reason, but also thinking if we should highlight the distinction b/w cont and disc? Esp since our result only works for continuous in really high dimensions}. %Hence, it is unclear whether these results generalize to other metric spaces, for instance, to doubling spaces. 
Feldman et al.~\cite{feldman2013turning} obtained 
coresets for
$k$-Means, PCA and projective clustering whose sizes are independent of the dimension.
Later, Sohler and Woodruff~\cite{sohler2018strong}
generalized these results and obtained coresets for $k$-Median and for subspace
approximation of sizes independent of the dimension $d$.
Braverman et al.~\cite{braverman2021coresets} designed coresets for  for $k$-Median in (the shortest-path metric of) an excluded-minor graph, which is also applicable to Euclidean metrics by using a terminal embedding result of Narayanan and Nelson~\cite{NarayananN19} that extends
the \emph{Johnson-Lindenstrauss Lemma}. After a long sequence of work~\cite{cohen2021near, DBLP:conf/stoc/HuangV20, DBLP:conf/icml/BakerBHJK020, braverman2021coresets, feldman2011unified}, finally in 2021, Cohen-addad, Saulpic, and Schwiegelshohn~\cite{cohen-addad-etal21:coreset-framework} gave a unified framwork for a large variety of metrics, ranging from Euclidean space, doubling metric, minor-free metric, as well as capturing the general metric case. However, while these results provide excellent bounds on coreset size, a key concern remains: the computational cost of clustering using them. Hence, a natural question to ask: \emph{Can coresets be used to obtain faster algorithms for $k$-Clustering?}

\paragraph*{EPASes for (Constrained) Clustering.} 
Efficient parameterized approximation scheme (EPAS)\footnote{A 
$(1+\varepsilon)$ approximation algorithm that runs in time $f(k,\varepsilon)\poly(n)$ for every $\varepsilon>0$.} is perhaps, the best combination of approximation and parameterized algorithms, and has been a majorly popular theme in recent algorithms research (see~\cite{AbbasiClustering23, bandyapadhyay2024parameterized, matouvsek2000approximate, ostrovsky2013effectiveness, ding2020unified} and the references therein). However, for $k$-Median/Means, unfortunately, it is not possible to obtain EPASes for points in general metric spaces~\cite{dasgupta08hardness-2-means, awasthi2015hardness}. Therefore, quite conceivably, much of the work focused on designing EPASes for points in special metrics, e.g., continuous Euclidean spaces~\cite{KumarSS10, badoiu-etal:approximate-clustering-coresets}. In a recent work~\cite{AbbasiClustering23}, Abbasi et al. 
presented a unified EPAS which works for various clustering objectives as well as diverse metric spaces. 
To that end, they introduced the notion of \emph{$\varepsilon$-scatter dimension} (see Section~\ref{sec:scatter-dim} for definition), which includes various commonly studied metric spaces, including, continuous Euclidean space, doubling spaces, etc. Very recently, Bourneuf and Pilipczuk~\cite{BourneufP25} showed that metrics induced by graphs from any fixed proper minor-closed graph class also have bounded scatter dimension. Moreover, they introduced metric analogs of well-known graph invariants from the theory of \emph{sparsity}, and as a consequence of this toolbox showed a \emph{coreset} for $k$-Center in any proper minor-closed graph class whose size is polynomial in $k$ but the exponent of the polynomial depends on the graph class and $\frac{1}{\varepsilon}$. Abbasi et al.~\cite{AbbasiClustering23} bypassed \emph{coresets} to obtain their clustering results. 
They mentioned that, obtaining clustering algorithms via coresets faces information-theoretic limitation, particularly in the context of $k$-Center, since coresets of desirable sizes do not exist in high-dimensional Euclidean spaces~\cite{braverman2019coresets}. However, a seemingly inherent limitation of their approach was that they could only handle clustering objectives, where each point was assigned to the closest chosen center, i.e., \emph{Voronoi clustering}\footnote{A recent work of Gadekar and Inamdar~\cite{DBLP:conf/stacs/Gadekar025} also used scatter dimension to design an EPAS for a common generalization of $k$-Center and $k$-Median that is not captured by the general norm framework of~\cite{AbbasiClustering23}; demonstrating that the notion of (algorithmic) scatter dimension has a wider applicability in clustering.
However, the clustering objective considered in \cite{DBLP:conf/stacs/Gadekar025} still can only handle \emph{Voronoi} assignments.}. %Ours is the first work that can exploit the properties non-Voronoi assignment constraints for $k$-Median/Means in metric spaces of bounded (algorithmic) scatter dimension. }. 
Indeed, in their work they explicitly asked the following open question:

\begin{tcolorbox}
\begin{question}\label{ques}
       \emph{Is it possible to design EPASes for constrained $(k,z)$-Clustering in metric spaces of bounded (algorithmic) scatter dimension, for assignment constraints such as capacities, fairness, and diversity?}
\end{question}
\end{tcolorbox}

In fact, \cite{AbbasiClustering23} asked this question for a more general problem, called constrained $k$-Clustering. However, this problem is so general that it also captures Capacitated $k$-Center, which is known to be W$[1]$-hard to approximate to a factor arbitrarily close to $1$ in metric spaces of bounded highway dimension~\cite{FeldmannV22}. Furthermore, it has been observed that metric spaces of bounded highway dimension have bounded scatter dimension\footnote{See the talk on \cite{AbbasiClustering23} at FOCS 2023: \url{https://focs.computer.org/2023/schedule/}.}, implying that the original question of~\cite{AbbasiClustering23} cannot be answered in the affirmative. Therefore, we reinterpret their question by restricting our focus to constrained $(k, z)$-clustering, where the  objective  is to minimize the sum of $z$th power of distances.
This problem is general enough to capture many well studied clustering problems (e.g., $z= 1, 2$, correspond to $k$-Median and $k$-Means, respectively), but also manages to bypass the hardness of result of~\cite{FeldmannV22}. 
}

\subsection{Scatter Dimension vs.\ Algorithmic Scatter Dimension}
Our algorithms work for  metric spaces with bounded scatter dimension.
At a high level, the scatter dimension\footnote{For simplicity, we focus on a particular metric space. These definitions can analogously be extended to a class of metric spaces; see~\Cref{sec:prelims}.} of a metric space $M=(P\cup F,d)$ is the length of the longest sequence of center--point pairs $(x_1,p_1),(x_2,p_2),\dots$, where $(x_i,p_i)\in F\times P$ such that (i) $d(x_i,p_i)>1+\epsilon$ for all pairs in the sequence, and (ii) $d(x_j,p_i)\le1$ for all $i<j$. Thus, this is a combinatorial property of a metric space that bounds the length of such sequences, which are called \emph{$\epsilon$-scattering sequences}.
Note that this notion neither relies on nor requires the existence of an algorithm
%\footnote{An algorithm that computes a center covering all previous points within unit distance is called a ball-intersection algorithm.}\todo{why do we need this footnote?} 
to compute such a sequence.
%, which is a computational task that itself may be hard.
%\todo[inline]{Is there any reason to believe that computing max length of such a sequence might be hard? It could be easy to compute/approximate, in that case reviews may get sidetracked, even though this is besides the point...}

On the other hand, the algorithmic scatter dimension of $M$ is the length of the longest $\epsilon$-scattering sequence computed by the \emph{best} ball-intersection algorithm, where ``best'' refers to an optimized procedure that produces the smallest possible length of such sequence for arbitrary choices of points. However, to design approximation algorithms under bounded scatter dimension, one still requires the existence of a ball-intersection algorithm for the metric space. In fact, for our framework, we require a stronger ball-intersection algorithm, namely $\mathcal{A}_B$ mentioned above, which not only satisfies the distance constraints but also returns a feasible center.

\subsection{Our Contributions}\label{ss: contri}
%\todo{I have polished this section significantly. Please check.}

%In this subsection, we start by outlining our main technical contribution, then describe various 

In this subsection, we begin by presenting an overview of our unified framework for constrained clustering problems, followed by a statement of the main result associated with this framework. We then discuss a range of applications of this result.

\paragraph*{A Unified Framework for Constrained Clustering.} In this work, we design a unified framework that yields EPASes for a large class of constrained  $(k, z)$-clustering problems, in metric spaces with bounded (algorithmic) scatter dimension. This, in particular, answers \Cref{ques} of Abbasi et al.~\cite{AbbasiClustering23} from FOCS 2023.
%
%we answer \Cref{ques} in the affirmative by designing a unified framework that yields EPASes for a large class of constrained  $(k, z)$-clustering in metric spaces with bounded algorithmic scatter dimension. 
%
% $k$-Median problems (in fact our framework naturally works for $(k, z)$-clustering, but we stick to $k$-Median in this overview). 
Towards this, we say that a center set $X$ of size $k$ is \emph{feasible} if it satisfies the center constraints Analogously, we say, for  a center set $X$  and  an assignment $f$, $(X,f)$ is \emph{feasible} if $X$ is  feasible and the assignment $f$ from  from $P$ to $X$ satisfies the assignment constraint. %Analogously, a center set $X$ of size $k$ is \emph{feasible} if there is a feasible assignment $f: P \to X$. 
% let $\mathcal{R}$ be the set of $k$ centers that
We require that the constrained $(k, z)$-clustering problem satisfies the following (somewhat informally stated) requirements. See Section~\ref{subsec:formalsetup} for the formal setup.
\begin{enumerate}
    \item There exists an algorithm $\mathcal{A}_C$ that computes a (client set) coreset $(Y, w)$ for the problem of size $\lr{\frac{k \log n}{\epsilon}}^{O(1)}$. \footnote{We emphasize here that, the required bound of $\lr{\frac{k \log n}{\epsilon}}^{O(1)}$ on the coreset size is attainable even in \emph{general metric spaces} for many natural constrained $(k, z)$-clustering problems; thus we do not need to rely on metric-specific properties to obtain coresets whose size is independent of $\log n$.
    %we \emph{do not require} coresets in the working metric space that has the bounded (algorithmic) scatter dimension in; even worse bounds obtained in general metrics. However, additional assumptions on the metric space can yield coresets of smaller sizes, which may lead to slightly improved running times; this will be discussed later.
    } %For our EPAS framework, $\mathcal{A}_C$ is even allowed to have FPT time. 
    \item There exists an algorithm $\mathcal{A}_A$ that takes as an input any weighted set of points $Z$ and candidate set of $k$ centers, $C$. Either it finds a near-optimal assignment from  all feasible  assignments from $Z$ to $C$, or outputs that no feasible assignment from $Z$ to $C$ exists.
    %We allow the algorithm to run in FPT$(k, \epsilon)$ time.
    %There exists an algorithm $\mathcal{A}_A$ that takes as an input a weighted point-set $(Z, w')$, a candidate set of centers $C$ of size $k$, and $\epsilon \in (0, 1)$, runs in FPT time, and either returns an $(1+\epsilon)$-approximate assignment from $Z$ to $C$ that satisfies all problem-specific assignment constraints; or correctly outputs that no feasible assignment from $Z$ to $C$ exists, and
    \item There exists an algorithm $\mathcal{A}_B$ that accepts $k$ requests, $Q_1,\dots,Q_k$, where $Q_i = \{(p, r): p \in Y, r \in \mathbb{R}\}$ and returns $X=\{x_1,\dots,x_k\}, x_i \in F$ such that  $X$ is feasible and for all $i\in[k]$, $x_i$ satisfies request set $Q_i$, i.e., $d(p,x_i)\le r$ for $(p,r)\in Q_i$,  or outputs $\bm{\bot}$ if no such center set exists. 
    
    % that takes as an input a finite set of \emph{requests} $Q = \{(p, r): p \in Y, r \in \mathbb{R}\}$, and an $\epsilon \in (0, 1)$, 
    % %runs in FPT$(k, \epsilon)$ time, 
    % and either returns a center $c \in F$ such that (i) $\d(p, c) \le (1+\epsilon) r$ for all requests $(p, r) \in Q$, and (ii) $c$ satisfies problem-specific feasibility requirements; or correctly outputs that no such center exists. 
\end{enumerate}
We call such a problem as \emph{\consprob} problem. We allow each of the three algorithms to run in time FPT in $k, \epsilon$, and potentially some problem-specific parameters. These requirements are extremely mild and are satisfied by many natural clustering problems. For instance, the capacitated $(k, z)$-clustering problem admits a coreset of the aforementioned size~\cite{Cohen-AddadL19Capacitated} in general metrics, its assignment algorithm can be implemented via a minimum-cost flow subroutine, and the algorithm $\mathcal{A}_B$ returns a feasible center with maximum capacity.

Our main result is summarized in the following theorem.

\begin{theorem}[Informal version of \Cref{thm:main}] \label{thm:introinformal}
    There exists a $(1+\epsilon)$-approximation for \consprob\ running in time $h(k, \epsilon) \cdot |\cI|^{O(1)}$ in metric spaces of bounded scatter dimension, where $|\cI|$ denotes the size of the input instance $\cI$. In addition, if the coresets for the problem can be computed deterministically, then the algorithm is deterministic.
\end{theorem}

As a corollary of this theorem, we obtain EPASes for several well-constrained $(k, z)$-clustering problems in metric spaces of bounded scatter dimension. These results are summarized in the first half of \Cref{tab:your_label}. In \Cref{subsubsec:scatterdim}, we state the results and outline the key ingredients used to obtain them. The formal proofs are given in Section~\ref{subsec:capacitated}, Section~\ref{subsec:fairclustering}, and Section~\ref{subsec:fault}, respectively.

In light of \Cref{ques}, it is natural to ask whether \Cref{thm:introinformal} can be extended to metric spaces of \emph{algorithmic} scatter dimension. While there are fundamental obstacles to achieving such an extension in full generality, which we discuss in \Cref{subsec:overview}, we are nevertheless able to obtain EPASes for certain specific instances of \consprob. These results are stated in the second half of \Cref{tab:your_label}, and in \Cref{subsubsec:algscatter}, with a detailed discussion provided in \Cref{ss:app:un}.

% in~\cite{AbbasiClustering23}

% In fact, for some soft constraints on center selection, we are able to extend our EPAS framework to metrics of bounded algorithmic scatter dimension.

\begin{table}[t]
    \centering
    \renewcommand{\arraystretch}{1.5}
    \begin{tabular}{|c|c|c|c|c|}
        \hline
        \rowcolor{lightgray} % Requires: \usepackage{colortbl}
        \textbf{Dimension}  & \textbf{Constraint} & \textbf{FPT Factor} & \textbf{Theorem} & \textbf{Section} \\
        \hhline{|=|=|=|=|=|}
        \multirow{3}{*}{\begin{tabular}[c]{@{}c@{}}Scatter Dimension \end{tabular}}  
          &
        \makecell{Capacities\\{\tiny uniform/non-uniform}} & \multirow{3}{*}{\begin{tabular}[c]{@{}c@{}} $\exp{(O(\frac{k \Lambda}{\epsilon} \log(\frac{1}{\epsilon}) \log(\frac{k \Lambda}{\epsilon}) )})$ * \end{tabular}}  & \multirow{3}{*}{\begin{tabular}[c]{@{}c@{}}Thm.~\ref{thm:capmatfault} \end{tabular}}    & \ref{subsec:capacitated}  \\
         \cline{2-2} \cline{5-5}
         & \makecell{Matroid} &  & & \ref{subsec:matroid} \\
         \cline{2-2}\cline{5-5}
         & \makecell{Fault Tolerant} &  &  & \ref{subsec:fault} \\
         \cline{1-5}         
           \multirow{2}{*}{\begin{tabular}[c]{@{}c@{}}Algorithmic Scatter\\  Dimension \end{tabular}}  
          & \makecell{$(\alpha, \beta)$-Fairness\\{\tiny Also: chromatic \& $\ell$-diversity}} & $\exp{(O(\frac{k \Gamma \Lambda}{\epsilon} \log(\frac{1}{\epsilon}) \log(\frac{k \Gamma \Lambda}{\epsilon}) )})$ * & Thm.~\ref{thm:fair} & \ref{subsec:fairclustering} \\
        \cline{2-5}
        % \makecell{Fault-tolerant} & $\exp{(O(\frac{k \Lambda}{\epsilon} \log(\frac{1}{\epsilon}) \log(\frac{k \Lambda}{\epsilon}) )})$ & Deterministic* & Thm.~\ref{thm:fault} & \ref{subsec:fault} \\
        % \hline        
         &  Vanilla & $\exp{(O(\frac{k \Lambda}{\epsilon} \log(\frac{1}{\epsilon}) \log(\frac{k}{\epsilon}) )})$  & Thm.~\ref{thm:vanilla} & \ref{subsec:vanilla}\\
        %\multirow{2}{*}{Vanilla} & $\exp{(O(\frac{k \Lambda}{\epsilon} \log(\frac{1}{\epsilon}) \log(\frac{k}{\epsilon}) )})$ & Randomized & \ref{thm:vanilla} & \ref{subsec:vanilla} \\
        %\cline{2-5}
        %& $\exp{(O(\frac{k \Lambda}{\epsilon} \log(\frac{1}{\epsilon}) \log(\frac{k \Lambda}{\epsilon}) )})$ & Deterministic* & \ref{thm:cap} & -- \\
        \hline
        %\hhline{|=|=|=|=|=|}
        %\hline
        %Outliers + X & $\lr{\frac{k+m}{\epsilon}}^{O(m)}$ overhead & Randomized & \cite{Jaiswal023,DabasGI25} & \ref{subsec:remarks} \\\hline
    \end{tabular}
    \caption{\small Consequences of our EPAS framework for \consprob. In the third column, only the FPT part of the running time is specified, and $\Lambda \coloneqq \lambda(\epsilon/c)$, where $\lambda(\cdot)$ is the (algorithmic) scatter dimension of the metric space, and $c$ is a constant that only depends on $z$. In the fourth row, $\Gamma$ represents the number of equivalence classes w.r.t.~ colors (see \Cref{def:fkm}). For the results marked as ``*'', \emph{if the coresets for the problem can be computed deterministically}, then our EPAS is completely deterministic. Note that the results of \cite{AbbasiClustering23} only work for the vanilla version, and the FPT factor of their randomized EPAS is $\exp(O(\frac{k \Lambda}{\epsilon} \log^2(\frac{k}{\epsilon}) ))$; thus our EPAS shaves off a $\log k$ factor from the exponent.}
    %In the last row, we get EPASes for the outlier versions of any of the other constraints with an overhead of a $\lr{\frac{k+m}{\epsilon}}^{O(m)}$ factor in the running time, where $m$ denotes the number of outliers.
    %\\Finally, all of our results also extend to $(k, z)$-clustering analogs where the objective involves $z$th power of distances, for a fixed $z \ge 1$ (all occurrences of $\epsilon$ need to be scaled by a factor of $z$).}
    \label{tab:your_label}
\end{table}

\subsubsection{EPASes for Bounded Scatter Dimension} \label{subsubsec:scatterdim}
Although, the problems mentioned here implicitly impose a soft constraint of selecting distinct centers in the solution, this can be circumvented by the method of color coding (see~\Cref{ss:app:con}). This technique preprocesses the input instance in FPT time such that the facilities have one of the $k$ colors and for a fixed optimal center set, all the centers have distinct colors. 
% The crux of the argument is that if a metric class $\cM$ has a bounded algorithmic scatter dimension, then any submetric of a metric from $\cM$ also has a bounded scatter dimension.

\paragraph*{EPASes for Capacitated,  Matroid, and Fault Tolerant $(k, z)$-Clustering:} As a consequence of our framework, we obtain the first EPASes for Capacitated $(k,z)$-Clustering in metric spaces of bounded  scatter dimension. Prior to this, EPASes for \emph{Capacitated $k$-Median/Means} were known only in Euclidean Spaces~\cite{Cohen-AddadL19Capacitated, BhattacharyaJK18}.
Furthermore, our results extend to  Matroid and Fault tolerant $(k, z)$-Clustering.
% ~\cite{KKNSS,KLS18}, where, we are additionally given a matroid $\mathbb{M}$ over $F$, and we additionally want that $X \subseteq F$ is an independent set in $\mathbb{M}$.  

%(also for fair versions~\cite{BandyapadhyayFS24,DBLP:conf/approx/BhattacharyaGJ21}). 

%Our result provides first EPASes for such general notion of metric spaces (Theorem~\ref{}). Prior to this, EPASes for \emph{Capacitated $k$-Median/Means} were known only in Euclidean Spaces~\cite{Cohen-AddadL19, BhattacharyaJK18} (also for fair versions~\cite{BandyapadhyayFS24,DBLP:conf/approx/BhattacharyaGJ21}), or for specific graph-based metrics~\cite{}. The framework of Abbasi et al.~\cite{AbbasiClustering23} provided a unified EPAS for various unconstrained clustering objectives in spaces with bounded scatter dimension, but extending such results to capacitated settings remained an open challenge. On the other hand, traditional coreset-based methods, which have been instrumental in designing FPT approximations for clustering~\cite{feldman2011unified, sohler2018strong, braverman2021coresets}, face inherent limitations in handling assignment constraints efficiently. By leveraging the structure of scatter dimension, we circumvent these barriers and develop EPASes for capacitated clustering, significantly expanding the frontier of algorithmic clustering beyond Euclidean settings. Our results further underscore the power of coresets as a key tool for designing efficient approximation algorithms in constrained clustering. 

\begin{restatable}{theorem}{capkm} \label{thm:capmatfault}
    For any fixed $z \ge 1$, the following problems admit an EPAS in metric spaces of  $\epsilon$-scatter dimension bounded by $\lambda(\epsilon)$ running in time $\exp{(O(\frac{k \Lambda}{\epsilon} \log(\frac{1}{\epsilon}) \log(\frac{k \Lambda}{\epsilon}) )} \cdot |\cI|^{O(1)}$, where  $\Lambda \coloneqq \lambda({\epsilon}/{c})$, and $c$ is a constant that only depends on $z$: (i) Capacitated $(k, z)$-Clustering, (ii)  Matroid $(k, z)$-Clustering, and (iii) Fault tolerant $(k, z)$-Clustering.
\end{restatable}

% \begin{restatable}{theorem}{matkm} \label{thm:mat}
%     For any fixed $z \ge 1$, Matroid $(k, z)$-Clustering admits an EPAS in metric spaces of  $\epsilon$-scatter dimension bounded by $\lambda(\epsilon)$ running in time $\exp{(O(\frac{k \Lambda}{\epsilon} \log(\frac{1}{\epsilon}) \log(\frac{k \Lambda}{\epsilon}) )} \cdot |\cI|^{O(1)}$, where  $\Lambda \coloneqq \lambda({\epsilon}/{c})$, and $c$ is a constant that only depends on $z$.
% \end{restatable}

% \begin{restatable}{theorem}{faultkm} \label{thm:fault}
%     For any fixed $z \ge 1$, Fault tolerant $(k, z)$-Clustering admits an EPAS in metric spaces of algorithmic  $\epsilon$-scatter dimension bounded by $\lambda(\epsilon)$ running in time $\exp{(O(\frac{k \Lambda}{\epsilon} \log(\frac{1}{\epsilon}) \log(\frac{k \Lambda}{\epsilon}) )} \cdot |\cI|^{O(1)}$, where  $\Lambda \coloneqq \lambda({\epsilon}/{c})$, and $c$ is a constant that only depends on $z$.
% \end{restatable}

\paragraph*{EPASes for metric spaces of bounded highway dimension.}

As referenced earlier, it has been observed that metric spaces of bounded highway dimension (BHD) also have bounded scatter dimension. Assuming this result, our framework implies an EPAS for \consprob~for any fixed $z \ge 1$ in BHD metric spaces; although we refrain from stating the result as a formal corollary since a formal proof of the claim that BHD metric spaces have bounded scatter dimension is not available in the public domain. 

We find this result interesting for a couple of reasons. Firstly, we are not aware of coresets of size $O_{\varepsilon, k}(1)$ for \consprob~in BHD metric spaces, i.e., coresets whose size is independent of $n$. Therefore, this would be the first EPAS for many instances of \consprob~in BHD metric spaces, to the best of our knowledge.  Secondly, our result in particular implies an EPAS for Capacitated $(k, z)$-Clustering in BHD metrics. This result is in stark contrast with Capacitated $k$-Center, for which the result of \cite{FeldmannV22} rules out such an EPAS (or even a PAS, i.e., an algorithm running in time $f(k) \cdot |\cI|^{g(\epsilon)}$) in BHD metrics. 

% We mention the following interesting corollary of our framework.

% \begin{corollary}
%     For a fixed $z \ge 1$,  \consprob\ admits an EPAS  in metrics of bounded highway dimension.
%     % and (2) capacitated \probname\ admits an EPAS in metrics of
%     % discrete Euclidean spaces of arbitrary dimension.
% \end{corollary}
% We find the aforementioned corollary interesting for the following reasons.
%We find this corollary interesting due to the following reason. As mentioned earlier, metrics of bounded highway dimension are contained in the class of metrics of bounded scatter dimension. Therefore, our framework implies EPASes for \consprob\ for metric spaces of bounded highway dimension. Specifically, this includes Capacitated $(k, z)$-Clustering, which is in stark contrast with Capacitated $k$-Center, for which the result of \cite{FeldmannV22} rules out such an EPAS (or even a PAS, i.e., an algorithm running in time $f(k) \cdot |\cI|^{g(\epsilon)}$). 

\subsubsection{EPASes for Bounded Algorithmic Scatter Dimension} \label{subsubsec:algscatter}
We now show applications of our framework to obtain EPAS for metrics of  bounded algorithmic scatter dimension for certain special cases.

\paragraph*{Fair \probname.}
Our framework extends to handle  $(\alpha,\beta)$-fair clustering\footnote{In this version, points are divided into multiple (potentially overlapping) groups, and we are additionally given upper and lower bounds on the number of points of each group that can belong to each of the $k$ clusters. The goal is to find a minimum-cost clustering of points satisfying these fairness constraints. A formal definition is given in \Cref{subsec:fairclustering}.}~\cite{BandyapadhyayFS24,BhattacharyaJK18}).  Again, this problem has a soft constraint of having distinct centers in $X$, but note that if the clusters corresponding to multiple copies of a center are ''fair", then they can be merged together by keeping just one copy of the center.
 
\begin{restatable}{theorem}{fairkm} \label{thm:fair}
     $(\alpha,\beta)$-Fair $(k, z)$-Clustering admits an EPAS in metric spaces of algorithmic $\epsilon$-scatter dimension bounded by $\lambda(\epsilon)$ running in time $\exp{(O(\frac{k \Gamma \Lambda}{\epsilon} \log(\frac{1}{\epsilon}) \log(\frac{k \Gamma \Lambda}{\epsilon}) )}) \cdot |\cI|^{O(1)}$, where  $\Lambda \coloneqq \lambda({\epsilon}/{c})$, and $c$ is a constant that only depends on $z$.
\end{restatable}

\paragraph*{Faster EPASes for Vanilla $(k, z)$-Clustering.} 
Since Vanilla (a.k.a. uncapacitated/Voronoi) $(k, z)$-Clustering is a special case of Well constrained $(k, z)$-Clustering, we obtain an EPAS for this variant with running time $\exp{(O(\frac{k \Lambda}{\epsilon} \log(\frac{1}{\epsilon}) \log(\frac{k \Lambda}{\epsilon}) )} \cdot |\cI|^{O(1)}$. We note that the running time of the EPAS from \cite{AbbasiClustering23} is $\exp(O(\frac{k \Lambda}{\epsilon} \log^2(\frac{k}{\epsilon}) )) \cdot |\cI|^{O(1)}$. Thus, note that our running time is already better than that of \cite{AbbasiClustering23} by a factor of $\log k$ in the exponent, but we incur an additional $\log \Lambda$ factor. However, in \Cref{subsec:vanilla}, we describe a simple modification to our framework, where we replace a key deterministic branching step by a randomized one, leading to an improvement in the running time. Specifically, compared to the EPAS of \cite{AbbasiClustering23}, our resulting running time \emph{matches} the dependence on $\epsilon$, and \emph{shaves off} a $\log k$ factor from the exponent. More formally, we prove the following theorem in \Cref{subsec:vanilla}.

\begin{restatable}{theorem}{vanillathm} \label{thm:vanilla}
    For any fixed $z \ge 1$,  $(k, z)$-Clustering admits an EPAS in metric spaces of algorithmic $\epsilon$-scatter dimension bounded by $\lambda(\epsilon)$ running in time $\exp{(O(\frac{k \Lambda}{\epsilon} \log(\frac{1}{\epsilon}) \log(\frac{k}{\epsilon}) )} \cdot |\cI|^{O(1)}$, where  $\Lambda \coloneqq \lambda({\epsilon}/{c})$, and $c$ is a constant that only depends on $z$.
\end{restatable}

\subsection{Technical Overview} \label{subsec:overview}
%Before we describe our techniques, we highlight key challenges on designing EPASes solely on algorithmic scatter dimension.

We start with an overview of the main ideas behind the EPAS of \Cref{thm:introinformal}, including the new challenges faced to handle \emph{constrained} clustering, as compared to prior work, and how we handle them. This is discussed in \Cref{subsubsec:techniques}. Then, in \Cref{subsubsec:algoscatchallenges}, we discuss why we cannot extend the result to \emph{algorithmic} scatter dimension. Finally, in \Cref{subsubsec:comparison} we compare our results with another orthogonal approach based on the work of \cite{Huang0L025}.

\subsubsection{Our Techniques} \label{subsubsec:techniques}
We briefly describe the technical ideas used to obtain the result, focusing mainly on constrained $k$-Median, for simplicity. Our framework is roughly follows the outline of the framework of~\cite{AbbasiClustering23} proposed  for Voronoi clustering problems.
On a high level, their framework is as follows. 
Recall that, in their problem, a point is assigned to the closest center, unlike our problem, where a point can be assigned to multiple centers. 
Let $\cC=\{\cC_1,\dots\cC_k\}$ be an optimal clustering of the coreset along with the corresponding centers $O=\{o_1,\dots,o_k\}$, and let $\opt$ be the optimal cost.
The main idea of their framework is to maintain a collection of $k$ request sets $(Q_1,\dots,Q_k)$ and a set $X=(x_1,\dots,x_k)$ of centers \emph{satisfying} the request sets. More specifically, each request set $Q_i$ is a collection of point-distance pairs $(p,\alpha) \in Y \times \mathbb{R}$, representing the demand that $p$ wants a center within distance $\alpha$ from itself. Additionally, the framework maintains an invariant that for $(p,\alpha)\in Q_i$, it must be that $p \in \cC_i$, i.e., $p$ belongs to $i^{th}$ cluster in the optimal clustering $\cC$.
The idea is to add requests to these request sets to force the algorithm to refine its solution $X$ until it becomes a $(1+\epsilon)$-approximate solution.
To argue this refinement, they show that if the cost of the current clustering $\cC'$ is larger than $(1+\epsilon) \cdot \opt$, then there must exist an \emph{unhappy}  point $p$:  $p \in \cC_i$, i.e., $p$ is assigned to $o_i$ in the optimal solution, such that $\d(p, x_i) > (1+\epsilon/3) \cdot \d(p, o_i)$, based on the fact that $p$ is assigned to the closest center. This gives a handle to refine $x_i$ by adding the request $(p,\tfrac{d(p,X)}{1+\epsilon})$ to $Q_i$, and by guessing $i \in [k]$ correctly.\footnote{One may view this as the algorithm trying to simulate $o_i$ by refining $x_i$.} Now, if the aspect ratio of the distances in $Q_i$ are bounded, then it can be shown that such a refinement can occur only bounded number of times, when the metric space has bounded algorithmic scatter dimension. Towards this, they introduce a concept of  \emph{upper bounds}  to initialize each $Q_i$, which again, is based on Vornoi-assignment.  
 However, we face multiple challenges when we try to extend their framework to our problems.
\begin{itemize}
    \item \emph{Upper bounds: } 
    As mentioned above, the upper bounds are inherently based on Vornoi properties,  and it is not clear how to design them for constrained problems.
    % \cite{AbbasiClustering23} introduce a notion of \emph{upper bounds}, and devise them for their problems, that help  bound the aspect ratio of radii in $Q_i$s. However, these upper bounds are inherently based on Voronoi properties, and it is not clear how to design them for our problems. This is where, we exploit the existence of general metric coresets. As mentioned before, coresets allows us to get a handle on the leaders and their distances, which in turn, allows us to bound the aspect ratio of radii in each request set, the task that upper bounds achieved in~\cite{AbbasiClustering23}.
    \item \emph{Assigning a point to multiple centers: } This distinction highlights a fundamental difference between our setting and the one considered in~\cite{AbbasiClustering23}. In their problem, a point $p$ is always assigned to a (closest) single center. This allows them to guess the index of its center $o_i$ in the optimal solution, and then use its current distance to refine $x_i$, as mentioned above.
    % This refinement then allows them to find a new center $x'_i$ which is substantially closer to $p$ than $x_i$. Finally, they argue that the number of such refinements are bounded  over all the points, when the metric space has a bounded (algorithmic) scatter dimension, leading to an EPAS for the problem. 
    On the other hand, in our setting, each point $p$ may be assigned to multiple centers.
    In fact, the bottleneck arises in a subtle way even when a point is assigned to only one (not necessarily a closest one) center. This is because, for constrained clustering, it is possible that $p$ is assigned to $x_j \in X$ in $\cC'$, while it is assigned to $o_i$ in $\cC$, for $i \neq j$, and it holds that $\d(p,x_j)>(1+\epsilon/3)\d(p,o_i)$ due to the cost of $(X,\cC')$. On the other hand, it can happen that $\d(p,x_i)$ is much smaller than $\d(p,x_j)$ (the worst case is when $\d(p,x_i)=0$, while $\d(p,x_j) > (1+\epsilon/3)\d(p,o_i)$),
which means that adding a request $(p,\tfrac{\d(p,x_i)}{1+\epsilon})$ does not refine $x_i$, as $x_i$ might have already satisfied the request. This can potentially make the algorithm run forever.
\end{itemize}

We now highlight our ideas to tackle these challenges.

\vspace{.1cm}
\noindent\emph{\textbf{Bounding distance aspect ratio of requests.}}
Towards this goal, we exploit the existence of small sized coresets for the problem:  we use the \emph{leader guessing} technique, which is a usual starting point in obtaining FPT-approximations for many clustering problems~\cite{Cohen-AddadG0LL19Tight,Cohen-AddadL19Capacitated}.
Let $(Y,w)$ be the coreset, and $(O, f^*)$ be its optimal solution.
More specifically, for each optimal center $o_i \in O$, we guess a point $p_i \in P$ with $f^*(p,o_i)>0$  and its approximate distance $\vartheta_i \approx \d(p_i,o_i)$\footnote{Approximate up to a factor of $(1+\epsilon)$.}, such  that $p_i$ is closest to  $o_i$ (ties broken arbitrarily). Since the size of the coreset is bounded, this guessing amounts to FPT many branches. 
Next, we add requests $(p_i,\vartheta_i)$ to $Q_i$ for all $i \in [k]$, and using algorithm $\mathcal{A}_B$, we compute a feasible set of centers $X$ satisfying all $Q_i$s. This means that, for each $i \in [k]$, center $x_i \in X$ satisfies $\d(p_i, x_i) \le (1+\epsilon)\d(p,o_i)$.
Finally, we use algorithm $\mathcal{A}_A$ to compute a near-optimal assignment $f$ of points in $Y$ to $X$.
This gives a baseline of minimum distances in each of the $k$ request sets maintained by the algorithm, and  hence, allows us to bound the distance aspect ratio of requests in each $Q_i$ by $O(1/\epsilon)$, whereas the loose upper bound approach of~\cite{AbbasiClustering23}  bounds it only by $O(k/\epsilon)$. This helps us in obtaining improved running time for uncapacitated $(k,z)$-clustering, as described above. 

\vspace{.1cm}
\noindent\emph{\textbf{Handling non-Voronoi assignments.}}
To tackle the multiple assignments to (non-Voronoi) centers issue, we consider a different line of argument. Specifically, we do the following. Let $(X,f)$ be the current solution whose cost is more than $(1+\epsilon)$ times the optimal cost. Let $(O,f^*)$ be a fixed optimal solution. 
Let $\hat{f}^*$ be the assignment of points in $P$ to $X$ based on $f^*$, i.e.,  $\hat{f}^*(p,x_i):=f^*(p,o_i)$, for all $i \in [k]$. 
Our crucial observation is that there must exist a pair $(p,i)$, called unhappy pair, such that $f^*(p,x_i)>0$, however $d(p,x_i) > (1+\epsilon/3)d(p,o_i)$; otherwise, this implies that the cost $(X,\hat{f}^*)$ is at most $(1+\epsilon)$ times optimal cost, and hence a contradiction to the cost of $(X,f)$, since $f$ is the best possible assignment on $X$. 
% Furthermore, these unhappy pairs can be guessed efficiently when we have coresets for the problem, and added to the request sets to refine the corresponding centers.
Since, the size of coreset $Y$ is small, we can  guess such unhappy pair $(p,i)$ using at most $|Y|\cdot k$ many branches.
% Then, we \emph{branch} over each of the $|Y|$ points of the coreset, and each of the indices in $[k]$, and 
Next, we add the request $(p, \frac{\d(p, x_i)}{1+\epsilon/3})$ to request set $Q_i$. Finally, we recompute $X$ using $\mathcal{A}_B$ to additionally satisfy the new request, and repeat the whole process, if the cost of the recomputed $X$ is still larger than $(1+\epsilon)\opt$.  
From the algorithmic point of view, the key step is to always maintain a set of centers that not only satisfies all the request sets, but also satisfies the additional constraints of the problem. For example, in the case of capacitated clustering, in every iteration,  we maintain a center $x_i$ that has the maximum capacity satisfying the request set $Q_i$, for every $i \in [k]$. This allows us to leverage the scatter dimension bound, while maintaining a feasible solution during the execution of the algorithm.
Furthermore, we show that these ideas allow us to bound the distances aspect ratio of requests in each $Q_i$ by $O(1/\epsilon)$. This, in turn, allows us to bound the number of requests that can be added to each $Q_i$ in metric spaces with bounded scatter dimension.
Thus, this leads to a recursive algorithm with branching factor $|Y|\cdot k $ and depth $O(\frac{k}{\epsilon} \log(\frac{1}{\epsilon})  \cdot \lambda(\frac{\epsilon}{c}))$, where $\lambda(\cdot)$ is the algorithmic scatter dimension of the metric space, and $c$ is a constant. 
Since $|Y| = \lr{\frac{k \log n}{\epsilon}}^{O(1)}$, this leads to the claimed running time by employing the standard trick, namely, $(\log n)^{t} = t^{O(t)} \cdot n^{O(1)}$.

\subsubsection{Challenges of Extending to Algorithmic Scatter Dimension} \label{subsubsec:algoscatchallenges}
Having presented the technical overview of our EPAS for metrics of bounded scatter dimension, we now revisit the issue of why this framework cannot be extended to metrics of bounded algorithmic scatter dimension. 
The notion of algorithmic scatter dimension is tied to a specific ball-intersection algorithm that returns a center satisfying a given set of distance constraints. Its boundedness guarantee applies only to this particular procedure; alternative algorithms may result in unbounded behavior.

In constrained clustering, however, center selection is subject to additional feasibility requirements. For example, capacitated clustering enforces global capacity constraints, while matroid constraints restrict the center set to be an  independent set. The ball-intersection algorithm underlying algorithmic scatter dimension is oblivious to such constraints, and its output may violate feasibility.
Furthermore, many constraints impose exclusivity or soft restrictions on center selection, such as selecting each facility at most once. Modifying the ball-intersection algorithm to respect these constraints may invalidate the boundedness guarantees of algorithmic scatter dimension.

These issues are intrinsic to constrained clustering. To overcome them, we base our framework on the purely combinatorial notion of $\varepsilon$-scatter dimension. This relaxation allows us to reason about arbitrary feasible center constructions while still exploiting bounded scattering behavior.
%As a result, we obtain EPASes for a broad range of constrained clustering problems in metric spaces of bounded scatter dimension, thereby answering \Cref{ques} positively. For certain special cases, we further extend our framework to metric spaces of bounded algorithmic scatter dimension by leveraging additional structural properties.

\subsubsection{Comparison to the work of  Huang et al.~\cite{Huang0L025}.} \label{subsubsec:comparison}  
We note that a recent work of Huang et al.~\cite{Huang0L025} gives a coreset framework for a number of constrained $(k, z)$-clustering objectives, including capacities and fairness, in several metric spaces. For continuous Euclidean, planar, and minor-free metrics, they design coresets for these objectives whose size is \emph{independent} of $\log n$. In \Cref{subsec:remarks}, we sketch an approach to obtain EPASes in the corresponding settings via essentially brute forcing on such coresets. However, this approach results in substantially worse running times as compared to our approach (e.g., $2^{\Omega_\epsilon(k^3 \log k)} \cdot \mathcal{I}^{O(1)}$ vs.~$2^{O_\epsilon(k \log k)} \cdot \mathcal{I}^{O(1)}$, see ~\Cref{subsec:capacitated} for Capacitated setting). Furthermore, although their framework designs coresets for many metric spaces that have bounded scatter dimension; it does not appear to give such a result for \emph{all metric spaces of bounded scatter dimension}. For example, metrics of bounded highway dimension (which have bounded scatter dimension), do not seem to be captured by their framework. Therefore, in our opinion, their result along with brute force enumeration
does not satisfactorily resolve \Cref{ques}.

\begin{table}[!ht]
    \centering
    \begin{tabular}{|c|c|c|}
    \hline
        \textbf{} & \textbf{This paper} & \textbf{Present}   \\ \hline
        FPT factor & \cellcolor[HTML]{EFEFEF} $2^{O_\epsilon(k\log k)}$ & $2^{O_\epsilon(k^3\log k)}$ \text{(follows from \cite{Huang0L025})} \\ \hline
    \end{tabular}
    \caption{\small EPAS running time  for Uniform capacitated $k$-Median in high dimensional Euclidean spaces.}\label{tab:capkmed}
\end{table}
%\newpage

\section{Preliminaries} \label{sec:prelims}

% \begin{itemize}
%     \item We will guess a value $\mathcal{G}$ such that $\opt \le \mathcal{G} \le (1+\sfrac{\epsilon}{9}) \cdot \opt$.
%     \item $Y$ may have co-located points.
%     \item Assignment $f: Y \times [k] \to \mathbb{N}$ is an assignment such that for any $p \in Y$, $\sum_{i = 1}^k f(p, i) = w(p)$. If for some $p \in Y$, $i \in [k]$, $f(p, i) > 0$, then we say that $f$ \emph{assigns} $p$ to $i$.
%     \item We fix an optimal solution $O = (o_1, o_2, \ldots, o_k)$ and an assignment $f^*$.
% \end{itemize}

%\section{Preliminaries}
We consider a metric (clustering) space $M=(P,F,\d)$, where $P$ is a set of $n$ points/clients, and $F$ is a (possibly infinite) set of facilities, and $\d$ is a metric over $P \cup F$. Let $\Delta_M$ be the distance aspect ratio of $M$; for brevity we use $\Delta$ instead of $\Delta_M$ when $M$ is clear from the context. Furthermore, a class $\mathcal{M}$ of metric spaces is a potentially infinite set of metric spaces. Note that, in case $F$ is infinite, the set $F$ may not be explicitly provided in the input. Additionally, points in $P$ have weights defined by the weight function $w: P \rightarrow \mathbb{R}_{\ge 0}$.

% \subsection{Notations} 
For a real $B> 1$ and $\delta >0$, we denote by $[B]_\delta$, the set containing $(1+\delta)^j$, for integers $0 \le j \le \frac{2 \log B}{\log (1+\delta)}$.

\subsection{Scatter Dimension}\label{sec:scatter-dim}
In this section, we revise useful facts related to  scatter dimension that were introduced in \cite{AbbasiClustering23}. For a comprehensive background on scatter dimension, see~\cite{AbbasiClustering23}. 

\begin{definition}[$\epsilon$-Scatter Dimension]\label{def:sd}
    Given a class $\cM$ of finite metric spaces, a metric space $M=(P,F,d) \in \cM$, and $\epsilon \in (0,1)$, an \emph{$\epsilon$-scattering sequence} in $M$ is a sequence $(x_1,p_1),\cdots, (x_\ell,p_\ell)$ of center-point pairs $(x_i,p_i) \in F \times P$ such that $d(x_i,p_i)> (1+\epsilon)$ for all $i \in [\ell]$, and $d(x_j,p_i)  \le 1$ for all $1\le i < j \le \ell$. The \emph{$\epsilon$-scatter dimension} of $M$ is the length of the longest $\epsilon$-scattering sequence present in $M$. The \emph{$\epsilon$-scatter dimension} of $\cM$ is the supremum of the $\epsilon$-scatter dimension over all $M \in \cM$.
\end{definition}

In this paper, we consider the following equivalent definition of $\epsilon$-scatter dimension, that is helpful in analyzing the algorithms.
\begin{definition}[$\epsilon$-Scatter Dimension (alternate definition)]\label{def:sdeq}
    Given a class $\cM$ of finite metric spaces, a metric space $M=(P,F,d) \in \cM$, and $\epsilon \in (0,1)$, an \emph{$\epsilon$-scattering sequence} in $M$ is a sequence  $(x_1,p_1,\alpha_1),\cdots, (x_\ell,p_\ell,\alpha_\ell)$ of center-point-radius triples $(x_i,p_i,\alpha_i) \in F \times P \times \mathbb{R}_{\geq 0}$ such that $d(x_i,p_i)> (1+\epsilon)\alpha_i$ for all $i \in [\ell]$, and $d(x_j,p_i)  \le r_i$ for all $1\le i < j \le \ell$.
    The \emph{$\epsilon$-scatter dimension} of $M$ is $\lambda_{M}(\varepsilon)$ if any $\epsilon$-scattering sequence contains at most $\lambda_{{M}}(\varepsilon)$ many triples per radius value. The \emph{$\epsilon$-scatter dimension} of $\cM$ is the supremum of the $\epsilon$-scatter dimension over all $M \in \cM$.
\end{definition}

\begin{definition}[Ball Intersection Problem]\label{def:ballint}
    In the \emph{Ball intersection problem}, we are given a metric space $M=(P,F,\d)$ from a metric class $\mathcal{M}$, a finite set  $Q \subsetneq P \times \mathbb{R}_{\ge 0}$ of point-distance constraints, called \emph{request set}, and a an error parameter $\eta >0$. The goal is to either (i) find a center $x \in F$ such that for all $(p,\alpha) \in Q$ if $\d(p,x) \le (1+\eta)\alpha$; or (ii) report  failure indicating that there does not exist a center $x^\star \in F$ such that for all $(p, \alpha) \in Q$, $\d(p, x^\star) \le \alpha$.

    We say $\mathcal{M}$ \emph{admits a ball intersection algorithm} if there is an algorithm that correctly solves the ball intersection problem for every metric space $M \in \mathcal{M}$, and runs in time polynomial in the size\footnote{When $F$ is infinite, then the size of $M$ is polynomial in $|P|$ and the space of storing a point.} of $M$ and $1/\eta$.
\end{definition}

\begin{definition}[Algorithmic $\varepsilon$-Scatter Dimension]\label{def:algsd}
	Let $\mathcal{M}$ be a class of metric spaces with ball intersection algorithm $\ballint_\mathcal{M}$. Let $M \in \mathcal{M}$ and $\varepsilon \in (0,1)$.  A \emph{$(\ballint_\mathcal{M},\varepsilon)$-scattering sequence} is a sequence $(x_1,p_1,\alpha_1),\dots,(x_\kappa,p_\kappa,\alpha_\kappa)$, where $\kappa$ is some positive integer, and for $i \in [\kappa], x_i \in F$, $p_i \in P$ and $\alpha_i \in \mathbb{R}_+$ such that
	\begin{itemize}
		\item $x_i = \ballint_\mathcal{M}(M,\{(p_1,\alpha_1),\dots,(p_{i-1},\alpha_{i-1})\},\varepsilon/2)$ \quad $\forall 2 \le i \le \kappa$
		\item  $\d(x_i,p_i) > (1+\varepsilon)\alpha_i$ \qquad $\forall i \in [\kappa]$
	\end{itemize}
	The \emph{algorithmic $(\varepsilon,\ballint_\mathcal{M})$-scatter dimension} of $\mathcal{M}$ is $\lambda_{\mathcal{M}}(\varepsilon)$ if any $(\ballint_\mathcal{M},\varepsilon)$-scattering sequence contains at most $\lambda_{\mathcal{M}}(\varepsilon)$ many triples per radius value. The \emph{algorithmic $\varepsilon$-scatter dimension} of $\mathcal{M}$ is the minimum algorithmic $(\varepsilon,\ballint_\mathcal{M})$-scatter dimension over any ball intersection algorithm $\ballint_\mathcal{M}$ for $\mathcal{M}$.
\end{definition}

The following lemma from~\cite{AbbasiClustering23} relates $\varepsilon$-scatter dimension of a class $\cM$ of finite metric spaces to its algorithmic $\varepsilon$-scatter dimension. The proof constructs a ball intersection algorithm that simply searches for a feasible center exhaustively. This allows us to bound $(\epsilon,\ballint)$-scatter dimension of $\cM$ for any ball intersection algorithm \ballint\ admitted by $\cM$. 
\begin{lemma}[\cite{AbbasiClustering23}]\label{lem:algvsnonalg}
    Every class $\cM$ of finite and explicitly given metric spaces with bounded $\varepsilon$-scatter dimension $\lambda(\epsilon)$ has algorithmic $\varepsilon$-scatter dimension $\lambda(\epsilon)$. In fact, for any ball intersection algorithm \ballint\ admitted by $\cM$, algorithmic $(\epsilon,\ballint)$-scatter dimension of $\cM$ is bounded by $\lambda(\epsilon)$. 
\end{lemma}

In this paper, we  design an EPAS framework for \probname\ in  metric spaces of bounded $\varepsilon$-scatter dimension. But the analysis of our framework (and algorithms) is based on its algorithmic $\varepsilon$-scatter dimension, which is also bounded due to the above lemma. Towards this, we will construct a suitable ball intersection algorithm that guarantees the constraints of the problem are satisfied. For this, we use the following lemma that is handy for bounding the number of iterations of our algorithm.
\begin{lemma}[\cite{AbbasiClustering23, DBLP:conf/stacs/Gadekar025}]\label{lem:len of scattering}
	For any class $\mathcal{M}$ of metric spaces with algorithmic $\varepsilon$-scatter dimension $\lambda(\varepsilon)$, there
	exists a \ballint\ algorithm $\mathcal{C}_{\mathcal{M}}$ with the following property. Given $\varepsilon \in (0,1)$, a constant $t\ge 1$, and $a_i >0, \tau_i\ge 2$ for $i \in [t]$, any $(\mathcal{C}_{\mathcal{M}}, \varepsilon)$-scattering contains $O(\sum_{i\in [t]}\lambda(\sfrac{\varepsilon}{2}) (\log\tau_{i})/\varepsilon)$ many triples whose radii lie in the interval $\cup_{i \in [t]}[a_i, \tau_i a_i]$.
\end{lemma}

\section{Framework}
In this section, we present our EPAS framework for constrained \probname. 

\subsection{Formal Setup of the Framework.} \label{subsec:formalsetup} 
In this paper, we consider \probname\ with \emph{assignment constraints}, called \emph{constrained $(k,z)$-Clustering}. The general setup is as follows. Let $z\ge 1$ be a fixed integer. We are given an instance $\cI$ of constrained $(k,z)$-Clustering over a metric space $M=(P,F,\d)$ along with weight function $w: P \rightarrow \mathbb{R}_{\ge 0}$,
% involves a metric space $\cM = (Z, \d)$, where $Z$ is a set of points,
% a set of clients $P \subseteq Z$ of size $n$, a set of candidate centers $F \subseteq Z$ , 
and an integer $k \ge 1$. 
% In addition, there may be  problem-specific constraints (such as capacities, fairness, etc.). 
The goal is to select a set $X \subseteq F$ of size $k$ and an assignment $f: P \times X \rightarrow \mathbb{R}_{\ge 0}$ with $\sum_{x \in X} f(p,x)=w(p)$ for every $p\in P$, ensuring that $(X,f)$ is feasible. We refer to such $X$ as a \emph{feasible} set of centers and
such an $f$ as a \emph{feasible assignment from $P$ to $X$}.  
The cost of pair $(X, f)$ --- which we also refer to as a solution---is defined as $\cost(P, X, f) \coloneqq \sum_{p \in P} \sum_{x \in X} f(p,x )\cdot \d(p, x)^z$. Furthermore, the cost of a feasible set of $k$ centers $X \subseteq F$ is defined as $\cost(P, X) \coloneqq \min_{f} \cost(P, X, f)$, where the minimum is taken over all feasible assignments $f$. If   $(f,X)$ is feasible, we refer to $(X,f)$ as a \emph{feasible solution}.
If no feasible assignment $f$ exists, or if $X$ is not feasible, then we define $\cost(P, X) \coloneqq \infty$. The goal is to find a feasible solution $(X, f)$ with minimum cost. In addition, our framework requires that the Constrained $(k,z)$-Clustering problem satisfies the following  properties.

\begin{description}[leftmargin=5mm]
    \item[Coresets.] Consider a subset of clients $Y \subseteq P$, and a weight function $w: Y \to \mathbb{N}$. For a $k$-sized subset $C \subseteq F$ of centers, we say that $f: Y \times C \to \mathbb{R}_{\ge 0}$ is a \emph{feasible} assignment from $Y$ to $C$ if (i) for each $p \in Y$, $\sum_{c \in C} f(p, c) = w(p)$, and (ii) $f$ satisfies additional problem-specific constraints\footnote{Note that this definition of a feasible assignment naturally generalizes the definition of an assignment from an unweighted set $P$ to $X$, given earlier: if $f: P \to X$ is a feasible assignment, then construct an assignment $f': P \times X \to \LR{0, 1}$ by letting $f'(p, c) = 1$ iff $f(p) = c$. }. We define the weighted-cost of $(Y, C, f)$ as $\wcost(Y, X, f) \coloneqq \sum_{p \in Y} \sum_{c \in C} f(p, c) \cdot \d(p, c)^z$. As before, we define the weighted-cost of $X$ as $\wcost(Y, C) \coloneqq \min_{f} \wcost(Y, C, f)$, where the min is taken over all feasible assignments $f: Y \times C \to \mathbb{N}$, and if no such feasible assignment exists, then $\wcost(Y, C, f) \coloneqq \infty$.
    
    For an $0 < \epsilon \le 1$, we say that $(Y, w)$ is an \emph{$\epsilon$-coreset} if, for each $C \subseteq F$ of size $k$, it holds that $|\wcost(Y, C) - \cost(P, C)| \le \epsilon \cdot \cost(P, C)$. 
    % Furthermore, if $(Y,w)$ with high probability. 
    
    %(I) for each valid assignment $f': P \to C$, there exists a valid assignment $f: Y \times C \to \mathbb{N}$, such that $|\cost(P, C, f') - \wcost(Y, C, f)| \le \epsilon \cdot \cost(P, C, f')$, and (II) for each valid assignment $f: Y \times C \to \mathbb{N}$, there exists a valid assignment $f': P \to C$, such that $|\cost(P, C, f') - \wcost(Y, C, f)| \le \epsilon \cdot \cost(P, C, f')$. The \emph{size} of an $\epsilon$-coreset $(Y, w)$ is defined as the number of distinct points in $Y$ (without weights).

    We require that, for any $0 < \epsilon < 1$, an $\epsilon$-coreset for the Constrained $(k,z)$-clustering problem of size $\lr{\frac{k \log n}{\epsilon}}^{O(1)}$ is computable in polynomial time\footnote{For randomized algorithms, we want that the algorithm outputs an $\epsilon$-coreset with high probability.}  in \emph{general metric spaces}. We also allow the size of the coreset to depend on some problem-specific parameters (e.g., the number of colors in \emph{fair} clustering), which will be reflected in the final running time. 
    
    \item[Efficient computation of (near-) optimal assignments.] There exists an algorithm,  \findassgn\ which takes as input a weighted point set $(Y, w)$, a candidate set $X=\{x_1,\dots,x_k\} \subseteq F$ of $k$ centers, and a parameter $0 < \epsilon < 1$, and returns a feasible assignment $f: Y \times X \to \mathbb{R}_{\ge 0}$ such that, for every $p \in Y$, it holds that $\sum_{j \in [k]} f(p,x_j) = w(p)$, and    $\wcost(Y, X, f) \le (1+\epsilon) \cdot \wcost(Y, X)$, where recall that $\wcost(Y, X)$ is the minimum-cost assignment from $Y$ to $X$. If no such assignment exists, then the algorithm returns $\bm{\bot}$. We say that $f$ \emph{assigns} $p \in Y$ to $x \in X$ if $f(p,x)>0$, and denote  $f(p,\cdot):= \{x \in X \vert f(p,x)>0\}$, the set of centers assigned to $p$ in $X$. 
    We require that the algorithm runs in time FPT in $k, \epsilon$, and other problem-specific parameters. Note that, when this algorithm is used with the original (unweighted) point set $P$ with unit-weight function $\mathbbm{1}$, it returns a feasible assignment $f: P \times X \to \LR{0, 1}$, which can be interpreted as a feasible assignment $f: P \to X$ in a natural way. Thus, it suffices to compute a set of centers $X$, since a near-optimal assignment from either coreset $(Y, w)$ or the original point set $P$ to $X$ can be computed using \findassgn\ algorithm at the overhead of FPT time. We denote by $\assgntime(\cI)$ as the running time of \findassgn\ on input instance $\cI$.\footnote{For simplicity, we hide the dependence of the running time on $\epsilon$.\label{foot:hideeps}}

% \item [Feasible]    
    \item[(Feasible) Ball intersection algorithm.] Let $Q = \LR{(p, r): p \in P, r \in \mathbb{R}_{\ge 0} }$ be a finite set of \emph{requests}.
    We say a center $x \in F$ \emph{satisfies} $Q$ if for each $(p, r) \in Q$, it holds that $\d(p, x) \le (1+\epsilon)r$.
    We assume that there exists an algorithm \genballint, which accepts $k$ requests, $Q_1,\dots,Q_k$, and returns $X=\{x_1,\dots,x_k\}, x_i \in F$ such that for all $i\in[k]$, $x_i$ satisfies request set $Q_i$, or outputs $\bm{\bot}$ if no such center set exists. In case there is a satisfying center set, then we refer by $\ballint_i, i \in [k]$ as the execution of \genballint\ that outputs $x_i \in X$, instead of the whole set $X$. Furthermore, \genballint\ satisfies the following conditions.
   
    % takes as an argument a set of requests $Q$ and $\epsilon \ge 0$, and either returns a center $x \in F$, such that for each $(p, r) \in Q$, $\d(p, x) \le (1+\epsilon)r$, or outputs $\bm{\bot}$ if no such center exists. In the former case, we say $x$ \emph{satisfies} $Q$.
    % For the general case, we assume a general \ballint\ algorithm, called \genballint, that satisfy the following.
    \begin{itemize}

        % \item Instead of having one request set, \genballint\  accepts $k$ requests, $Q_1,\dots,Q_k$, and returns $X=\{x_1,\dots,x_k\}, x_i \in F$ such that for all $i\in[k]$, $x_i$ satisfies equest set $Q_i$. We refer by $\ballint_i, i \in [k]$ as the execution of \genballint\ that outputs $x_i$, instead of $X$.
        % is computed using \ballint\ on the request set $Q_i$.

        \item     The returned center set $X$ must be a feasible.
        % to satisfy certain additional problem-specific properties so that $X$ is feasible for the problem. 
        \item \emph{Compatibility with \findassgn.} Let $X$ be set of $k$ centers obtained from \genballint\ satisfying request sets $Q_1,\dots,Q_k$, and let $f$ be a near-optimal assignment obtained from \findassgn\ subroutine for $X$. Then, for every feasible solution $(Y,g)$  for the problem, there exists a bijection $b: Y \leftrightarrow X$ such that that $(X,g')$ is also a feasible solution, where $g'(p,x):=g(p,f(y))$, for $p \in P, x \in X$. 
    \end{itemize}

    We require that this algorithm runs in time FPT in $k, \epsilon$, and the scatter dimension of the metric space.  We denote by $\ballinttime(\cI)$ as the running time of \genballint\ on input instance $\cI$.$^{\ref{foot:hideeps}}$
\end{description}

We call such constrained $(k,z)$-Clustering problem as \emph{\consprob\ problem}. We allow each of the three algorithms to run in time FPT in $k, \epsilon$, and potentially some problem-specific parameters. We use $|\cI|$ to denote the size of the instance.\footnote{When $F$ is infinite, then $|\cI|$ is polynomial in input parameters other than $F$ and the space of storing a point.}

% When $F$ is infinite, we denote
% As we shall see later, these requirements are extremely mild and are satisfied by many natural clustering problems.

\subsection{Algorithm} \label{subsec:epas}

Our main result is the following.
% \agi{scale $\lambda$ appropriately.}
\begin{restatable}{theorem}{mainthm} \label{thm:main}
Let $\cM$ be a class of metric spaces with algorithmic $\epsilon$-scatter dimension bounded by $\lambda(\epsilon)$. 
There is a EPAS that given an instance $\cI$ of \consprob\ over a metric space in $\cM$ runs in time $\epastime$, where $\coresetsize$ is the size of coreset for $\cI$. Furthermore, the running time of the EPAS is $\exp{(O(\frac{k \lambda(\sfrac{\epsilon}{60z})}{\epsilon}) \log(\frac{1}{\epsilon}) \log(\frac{k \lambda(\sfrac{\epsilon}{60z})}{\epsilon}) )}) \cdot |\cI|^{O(1)}$
%$\epastimecoreset$
%
%\todo{maybe we shouldnt hide polylog factors, since we are stating it as our contribution. also the coreset running time should have $\log(\lambda(\epsilon/60))$ factor, which is not hideable in tilde}\ag{do you want to do this?}, 
when $\ell=(\tfrac{k}{\epsilon}\log|\cI|)^{O(1)}$.
\noindent Additionally, if coresets  for $\cI$ can be constructed deterministically, then the EPAS is deterministic.
\end{restatable}

The pseudo-code of our main algorithm yielding \Cref{thm:main} is described in~\Cref{algo:main}. This algorithm is based on the ideas mentioned in~\Cref{ss: contri}, and it uses the recursive algorithm~\Cref{algo:apx} as a subroutine. \Cref{algo:main}
takes as input, an instance $\cI$ of \consprob, a ball intersection algorithm \ballint, an assignment algorithm \findassgn, and a guess $\cG$ of optimal cost. Note that the guess, $\cG$ of optimal cost, can be computed efficiently, up to a multiplicative factor of $(1+\epsilon)$.
The main task of the subroutine \Cref{algo:apx} is to compute a feasible solution $(X,f)$ for the set of requests $(Q_1,\dots,Q_k)$ it receives, where $X$ is a set of centers and $f$ is a feasible assignment of clients to those centers. If the resulting solution has cost close to the guessed optimum, the algorithm terminates and returns it. Otherwise, it identifies an unhappy pair $(p,i)$ with respect to $(X,f)$, modifies the request set $Q_i$ accordingly, and recurses.
\begin{algorithm}[h]
	\caption{Approximation Scheme for \consprob}
	\label{algo:main}
	\begin{algorithmic}[1]
          % \KwData{An instance $\cI$ of constrained $(k,z)$-Clustering, $\epsilon>0$}
		\Statex \textbf{Input:} Instance $\cI$ of \consprob\ in metric space $M=(P,F, \d) \in \cM$, $\varepsilon \in (0, 1/2]$,  \text{Ball intersection} algorithm \ballint, \findassgn\ Algorithm and   $\opt \le \mathcal{G} \le (1+\epsilon) \cdot \opt$.
		\Statex \textbf{Output:} A feasible solution $(S \subseteq F,f_S)$ such that $\cost(P,S,f_S) \le (1+30\varepsilon) \opt$.    
		\State $(Y, w) \gets$ coreset for $\cI$ \label{algo:main:coreset}
        \For{every tuple $(p'_1, p'_2, \ldots, p'_k) \in Y^k$ and  distances $(r'_1, r'_2, \ldots, r'_k) \in [\Delta]_{\epsilon}^k$}\label{algo:main:for}
            % \State \textbf{for} $i \i [k]$, let $Q_i = \LR{(p'_i, r'_i)}$ and 
            \State Let $Q_i = \LR{(p'_i, r'_i)}$ for all $i \in [k]$\label{algo:main:initQ}
			\State Let  $\cB=\{B_1,\dots, B_k\}$, where $B_i = \ball(p'_i,\sfrac{\c r'_i}{\epsilon})$ 
            % \{p\in Y : \sfrac{r'_i}{(1+\epsilon/9)} \le \d(p,p'_i) \le \sfrac{\c r'_i}{\epsilon}\}$
            % \LR{q \in Y:  \d(q, p'_i) \le \frac{\c r'_i}{\epsilon}}$. 
            % \State \Return $\epasalgo((Q_1,\dots, Q_k),\cB)$ if not $\bot$
            % \State If $(X',f') \neq (\bot,\bot)$ then \Return $(X',f')$ 
            \State $S \gets \epasalgo((Q_1,\dots, Q_k),\cB)$\label{algo:main:call}
            \State \textbf{If} $S \neq \bot$ \textbf{then} \Return $(S,\findassgn(P,S,\epsilon))$\label{algo:main:ret}
     
        \EndFor
        \State \Return $\bot$
	\end{algorithmic}
\end{algorithm}
\begin{algorithm}[h]
	\caption{$\epasalgo((Q_1,\dots,Q_k), \cB)$}
	\label{algo:apx}
	\begin{algorithmic}[1]
        \State Let $X = (x_1,\dots,x_k) \gets \genballint(Q_1,\dots,Q_k)$ \label{algo:apx:ballint}
        % such that $x_i = \ballint(Q_i,\epsilon), i \in [k]$
        \Statex \Return $\bot$ if no such $X$ was found
        \State $f\gets \findassgn(Y,X,(Q_1,\dots,Q_k),\epsilon)$, \Return $\bot$ if it fails
        \State \textbf{If} {$\wcost(Y, X, f) \le (1+5\epsilon) \cdot \mathcal{G}$} \textbf{then} \Return $X$
        \For{every index $i \in [k]$ and point $p \in B_i \setminus \ball(x_i,r'_i)$}\label{algo:apx:for}
            \State \Return $\epasalgo((Q_1,\dots, Q_i \cup \{(p,\frac{\d(p,x_i)}{1+\epsilon})\},\dots, Q_k),\cB)$ if not $\bot$\label{algo:apx:addQi}
            % \State $(X,f) \gets \epasalgo((Q_1,\dots, Q_i \cup \{(p,\frac{\d(p,x_i)}{1+\epsilon/3})\},\dots, Q_k))$\label{algo:apx:addQi}
            % \State If $(X,f) \neq \bot$ then \Return $(X,f)$ 
        \EndFor
        \State \Return $\bot$
	\end{algorithmic}
\end{algorithm}

\subsection{Analysis}
In this section, we analyze the running time and correctness of~\Cref{algo:main}.
First, in~\Cref{ss:runtime}, we show that~\Cref{algo:main} has bounded running time (\Cref{lem:main time}), and later, in~\Cref{ss:correctness}, we show that  there exists an iteration of \textbf{for} loop in~\Cref{algo:main:for} of~\Cref{algo:main} that does not return $\bot$ (\Cref{lem:correctness}). Finally, we use these results to prove~\Cref{thm:main}. 

For keeping the analysis clean, we show the proof for \conskmed. We remark that the same analysis can be extended easily to \consprob, by replacing $\epsilon$ with $\epsilon/z$.
Let $\cI$ be an instance of \consprob\ that is given as an input to~\Cref{algo:main}, and let $(Y,w)$ be the corresponding coreset for $\cI$ obtained in~\Cref{algo:main:coreset}. 
\subsubsection{Running time}\label{ss:runtime}
In this section, we bound the running time of~\Cref{algo:main}.
\begin{lemma}\label{lem:main time}
    \Cref{algo:main} terminates within time $\algotime$.
\end{lemma}
For an execution of $\epasalgo((Q_1,\dots,Q_k), \cB)$, we call the first argument $(Q_1,\dots,Q_k)$ of the execution as the  \emph{request set} of the execution.
Towards this, we show the following claim.
\begin{claim}\label{cl:rec:depth}
    The depth of the recursion of~\Cref{algo:main} is at most $\pathlen$.
\end{claim}
\begin{proof}
Consider the execution of an iteration of the \textbf{for} loop of~\Cref{algo:main} with tuple $(p'_1,\dots,p'_k)$ and distances $(r'_1,\dots,r'_k)$ and the corresponding invocation of \epasalgo\ with arguments $(Q_1,\dots,Q_k)$ and $\cB$. Let $\cT$ be the corresponding recursion tree with root node $t_0$. For any node $t \in \cT$, we denote by $(Q^t_1,\cdots,Q^t_k)$ as the request set with which it is invoked during the execution.
Furthermore, we call the request set corresponding to root node $t_0$ as the \emph{initial} request set of $\cT$.
% Let $\cT$ be the recursion tree corresponding to this particular invocation of \epasalgo. 
Notice that the request set of any non-root node $t \in \cT$ differs with the request set of its parent node  $t'$ in exactly one request. More specifically, there exists $j \in [k]$ such that $Q^t_{j} = Q^{t'}_{j} \cup \{(p,\alpha) \}$, for some $(p,\alpha) \in Y \times [\Delta]_{\epsilon}$, and $Q^t_i = Q^{t'}_i$, for $i \neq j$. 
% iteration of \epasalgo, we recursively invoke \epasalgo\ by adding a request to exactly one $Q_i, i\in[k]$.
In this case, we say that $t$ has a \emph{modified} $Q_j$, and that $(p,\alpha)$ as the \emph{modifier} of $t$.
% , for $i\in[k]$,  if $t$ is invoked by its parent (in $\cT$)  by adding a request to $Q_i$.
Consider a path $\cP$ from the root of $\cT$, and for $i \in [k]$, let $T_i = \{t^1_i,t^2_i,\dots\} \subseteq \cP$ be the nodes in $\cP$ that have modified $Q_i$.
% Let $Q_i^j, j \in [|T_i|]$ be the cluster constraint at the start of node $t^j_i \in T_i$, let $(p^j_i,\alpha^j_i)$ be the request 
Let  $x^1_i,x^2_i,\dots$ be the centers computed for cluster $i$  at nodes $t^1_i,t^2_i,\dots$, respectively (see~\Cref{algo:apx:ballint}), and let
$(p^1_i,\alpha^1_i),(p_i^2,\alpha^2_i),\dots$ be the modifiers of
% requests added to $Q_i$  at 
nodes $t^1_i,t^2_i,\dots$, respectively (see~\Cref{algo:apx:addQi}). That is, $x^j_i$ is the center computed  at $t^j_i \in \cP$ for cluster $i$. Since, for $j>1$, node $t^j_i$ computes
$X^j=(x^j_1,\dots,x^j_k)$ using \genballint\ such that $x^j_i$  satisfies the cluster constraint  containing $(p^1_i,\alpha^1_i),\dots,(p_i^{j-1},\alpha^{j-1}_i)$ requests, i.e., it holds that $\d(x^j_i,p^{j'}_i) \le (1+\epsilon)\alpha^{j'}_i$, for $j' < j$. Furthermore, since node $t^j_i$ adds $\left(p^j_i, \frac{\d(p^j_i,x^j_i)}{(1+\epsilon)}\right)$ to $Q_i$, it holds that $\alpha^j_i = \frac{\d(p^j_i,x^j_i)}{(1+\epsilon)}$. Hence, $\d(p^j_i,x^j_i)> (1+\tfrac{\epsilon}{2})\alpha^j_i$.
Therefore, the sequence $(x^1_i,p^1_i,\alpha^1_i),(x^2_i,p^2_i,\alpha^2_i),\dots$ is a $(\ballint_i, \tfrac{\epsilon}{2})$-algorithmic 
% $\tfrac{\epsilon}{2}$
scattering. Finally, since $p^j_i \in \ball(p'_i,\sfrac{\c r'_i}{\epsilon}) \setminus \ball(x_i,r'_i)$, where $(p'_i,r'_i)=Q^{t_0}_i$ is the $i^{th}$ request in the initial request set of $\cT$, we have that, $\alpha^j_i=\frac{d(p^j_i,x^j_i)}{1+\epsilon} \ge r'_i$. On the other hand,
    \[
    \alpha^j_i=\frac{d(p^j_i,x^j_i)}{1+\epsilon} \le \frac{d(p^j_i,p'_i) + d(p'_i, x^j_i)}{1+\epsilon} < \frac{\c r'_i}{\epsilon} + d(p'_i, x^j_i) \le \frac{7r'_i}{\epsilon},
    \]
where the last inequality follows since $ d(p'_i, x^j_i) \le (1+\epsilon)r'_i$, due to the initial request $(p'_i,r'_i) \in Q^{t_0}_i$. Thus, the radii in the requests  $(p^1_i,\alpha^1_i),(p_i^2,\alpha^2_i),\dots$ lie in the interval $\left[{r'_i},{7r'_i}/{\epsilon}\right]$. Hence, using~\Cref{lem:len of scattering}, the length of  $(\ballint_i, \epsilon/2)$-algorithmic scattering $(x^1_i,p^1_i,\alpha^1_i),(x^2_i,p^2_i,\alpha^2_i),\dots$ is bounded by $\sclen$, since the $(\ballint_i, \epsilon/2)$-algorithmic dimension of $\cM$ is bounded by $\lambda(\epsilon)$ due to~\Cref{lem:algvsnonalg}. Since every node in $\cP$ has a modified $Q_j, j\in[k]$, we have that the length of the path $\cP$ is bounded by $\pathlen$.
\end{proof}

Now, we bound the running time of~\Cref{algo:main} as follows.

\noindent\textit{Proof of~\Cref{lem:main time}.}
The \textbf{for} loop (\Cref{algo:main:for}) runs for at most $(\coresetsize\cdot [\Delta]_\epsilon)^k$ iterations. Since, every iteration of the \textbf{for} loop generates a recursion tree with depth $\pathlen$ due to~\Cref{cl:rec:depth}, and each node in the tree has $k\coresetsize$ many children, the total running time of~\Cref{algo:main} is bounded by
$
(\coresetsize [\Delta]_\epsilon)^k \cdot (k\coresetsize)^{\pathlen}\poly(|\cI|)\cdot \assgntime(\cI)\cdot\ballinttime(\cI)$, which is $\epastime,
$
since $\Delta = \poly(|\cI|)$.
$\hfill\square$

\subsubsection{Correctness}\label{ss:correctness}
In this section, we show that~\Cref{algo:main} never returns  $\bot$. 
Specifically, we show the following guarantee.% about~\Cref{algo:main}.
\begin{lemma}\label{lem:correctness}
% \Cref{algo:main} never returns $\bot$.
     There exists an iteration of \textbf{for} loop (\Cref{algo:main:for}) of~\Cref{algo:main} whose \epasalgo\ (\Cref{algo:main:call})  never returns $\bot$.
     % is at least one path in the \textcolor{red}{execution} tree of~\Cref{algo:main} such that no node in that path returns $\bot$.
\end{lemma}
\begin{proof}
 For any feasible solution $(X,f)$, we say that $f$ \emph{assigns} $p \in Y$ to $x \in X$ if $f(p,x)>0$, and denote  $f(p,\cdot):= \{x \in X \vert f(p,x)>0\}$, the set of centers assigned to $p$ in $X$. 
Let $\cO=(O=(o_1,\dots,o_k),f^*)$ be an optimal solution to coreset $(Y,w)$ with cost $\opt$\footnote{For brevity, we abuse \opt\ to denote the cost of an optimal solution in coreset rather than that of the original instance. This is fine, since the cost of optimal solution in coreset is same as that in the input instance, up to a multiplicative factor of $(1+\epsilon)$.}.
For $i\in [k]$, let $p^*_i= \arg\min_{p \in P} \{d(p,o_i) \mid f^* \text{ assigns $p$ to $o_i$}\}$, be the \emph{leader} of (the cluster corresponding to) $o_i$, by  breaking ties arbitrarily. Similarly, let $r_i=d(p^*_i,o_i)$ be the distance of $o_i$ to its leader $p^*_i$. Now, we fix the iteration of the \textbf{for} loop of~\Cref{algo:main} corresponding to $(p'_1,\dots,p'_k)=(p^*_1,\dots,p^*_k)$ and $(r'_1,\dots,r'_k)=(\tilde{r}_1,\dots,\tilde{r}_k) \in [\Delta]^k_{\epsilon}$, where $r_i \le \tilde{r}_i< (1+\epsilon)r_i, i \in [k]$. Let $\cT$ be the recursion tree of \epasalgo\ corresponding to this iteration with root $t_0$. 
\begin{observation}\label{obs:rootnode}
    At the root node $t_0$ of $\cT$, we have that,  for $i \in [k]$, $Q_i = (p^*_i,\tilde{r}_i)$, such that  $r_i \le \tilde{r}_i< (1+\epsilon)r_i$. 
\end{observation}

We will show that there is at least one path  $\cP$ in $\cT$ such that no nodes in $\cP$ return $\bot$. Towards this, we define the following notion of consistency for a node in $\cT$.
\begin{definition}[Consistency]
	% Consider a fixed hypothetical optimal solution $O=(o_{1},\cdots,o_{k})$.
	A node $t \in \cT$ with first argument  $(Q_1,\dots,Q_k)$  is said to be \emph{consistent} with $O$ if for any request $(p,\alpha)\in Q_{i},i\in[k]$,
	we have $\d(p,o_{i})\le \alpha$.
\end{definition}

Notice that if a node $t \in \cT$ is consistent with $O$, then it successfully finds $X$ using \ballint\ in~\Cref{algo:apx:ballint}, since $o_i$ always satisfies all the requests in $Q_i, i\in [k]$.  Furthermore, since $f^*$ is feasible, \findassgn\ also finds a feasible assignment.
Our aim, now, is to show that there is at least one path $\cP$ in $\cT$ such that all nodes in $\cP$ are consistent. This means that none of the nodes in $\cP$ return $\bot$, as desired.

Towards this, note that the root node $t_0 \in \cT$ has  first argument as $((p'_1,\tilde{r}_1),\dots,(p'_k,\tilde{r}_k))$ and hence, is consistent with $O$, since for $(p'_i,\tilde{r}_i) \in Q_i$, we have that $\d(p'_i,o_i) =r_i \le  \tilde{r}_i$. Now, we show inductively that given a non-leaf node $t \in \cT$ that is consistent with $O$, there exists at least one child $t'$ of $t$ that is consistent with $O$.
% Since $t$ is consistent with $O$, 
Let $(X,f)$ be the solution computed by node $t$ (such solution exists since $t$ is consistent with $O$). 
% We say a point $p \in Y$ is an \emph{$\epsilon/9$-witness} to cluster $i \in [k]$ if $\d(p,x_i) > (1+\epsilon/9) \d(p,o_i)$. Similarly, we say that a cluster $i \in [k]$ is a \emph{witness} to $(X,f)$ if there is an $\epsilon/9$-witness to $i$.
Then, the following claim is a key for consistency of $t'$.

\begin{restatable}{claim}{witnesslemma} \label{lem:witnesslemma}
	If $\wcost(Y, X, f) > (1+5\epsilon)\cdot \mathcal{G}$, then there exists a point $p \in Y$ and an index $i \in [k]$ satisfying the following properties.
	\begin{enumerate}[label=(\alph*)]
		\item $f^*$ assigns $p$ to $i$, 
		\item $\d(p, x_i) > (1+\epsilon) \cdot \d(p, o_i)$,
		\item $\d(p, x_i) \ge \tilde{r}_i $, and 
		\item $\d(p, p^*_i) \le \frac{\c }{\epsilon} \tilde{r}_i$.
	\end{enumerate}
\end{restatable}
\begin{proof}
    % \textcolor{red}{\small AG: new take.}
    % For $p \in Y$, let $f^*(p,\cdot):= \{i \in [k] \vert f^*(p,x_i)>0\}$, furthermore, 
    Let $W\coloneqq \{(p,i) \in Y \times [k]: i \in f^*(p,\cdot) \text{ and } \d(p, x_i) > (1+\epsilon) \cdot \d(p, o_i)\}$.
        % Let $W \coloneqq \LR{(p, i) \in Y \times [k]: \text{$f^*$ assigns $p$ to $i$ and } \d(p, x_i) > (1+\epsilon) \cdot \d(p, o_i)}$. 
    First, we show that $W$ is non-empty. Suppose for the sake of contradiction that $W$ is empty. Since, $\sum_{j \in [k]}f^*(p,o_j) = \sum_{j \in [k]}f^*(p,x_j)  =w(p)$, this means that, for each $p \in Y$ and all $i \in [k]$ such that $f^*$ assigns $p$ to $o_i$, it holds that $\d(p, x_i) \le (1+\epsilon) \cdot \d(p, o_i)$. 
    Then, by the properties of \genballint\ and \findassgn\ algorithms,
    and due to their compatibility, it follows that $(X,f^*)$ is a feasible solution. Furthermore,  for all $p \in Y$,  we have,
    \begin{align*}
        \wcost(p,X,f^*) &= \sum_{i \in f^*(p,\cdot)} f^*(p,x_i) \cdot d(p,i) \\
        &\le \sum_{i \in f^*(p,\cdot)} (1+\epsilon)\cdot d(p,o_i)\\
        &= (1+\epsilon)\cdot\wcost(p,O,f^*)
    \end{align*}
    This means, $\wcost(Y,X,f^*) = \sum_{p \in Y} \wcost(p,X,f^*) \le (1+\epsilon)  \sum_{p \in Y} \wcost(p,O,f^*)=(1+\epsilon) \cdot \wcost(Y,O,f^*)$.
     Therefore, we have
   \begin{align*}
       \wcost(Y, X, f) &\le (1+\epsilon)\cdot\cost(Y, X, f^*) & \text{due to \findassgn}\\
       &\le (1+\epsilon)^2 \cdot \wcost(Y,O,f^*)\\
       &\le (1+\epsilon)^3 \cdot \wcost(P,O,f^*) & \text{due to coreset}\\       
       &\le (1+\epsilon)^3 \cdot \cG & \text{since } \cG \ge \opt\\
       &\le (1+5\epsilon)\cdot \cG & \text{since } \epsilon\le 1/2
   \end{align*}
    % Therefore, the current assignment $f$ satisfies that $\cost(Y, X, f) \le (1+\epsilon)\cdot\cost(Y, X, f^*) \le (1+\epsilon)^2 \cdot \opt \le (1+\sfrac{\epsilon}{3}) \mathcal{G}$, contradicting the assumption. 
    This contradicts the premise of the claim.    Therefore, $W \neq \emptyset$. 
	
	Next, consider any pair $(p, i) \in W$, and note that, properties $(a)$ and $(b)$ are satisfied for $p$ and $i$, by definition. Next we show that $(p,i)$ also satisfy properties $(c)$ and $(d)$.
    For $(c)$, note that,
    \[
    \d(p,x_i) > (1+\epsilon)\cdot \d(p, o_i) = (1+\epsilon)r_i > \tilde{r_i},
    \]
    as desired, since  $r_i \le \tilde{r}_i < (1+\epsilon)r_i$ due to~\Cref{obs:rootnode}.   
    Now, suppose the property $(d)$ does not hold, i.e., $\d(p, p^*_i) > \frac{\c \tilde{r}_i}{\epsilon}$. 
    % Let $o_i \in O$ denote the optimal center.  
	Then, we obtain the following by triangle inequality.
	\begin{equation}
		\d(p, x_i) \le \d(p, p^*_i) + \d(p^*_i, x_i) \le \d(p, p^*_i) + (1+\epsilon) \tilde{r}_i \le (1+\tfrac{\epsilon}{4}) \cdot \d(p, p^*_i) \label{eqn:currentcenter}
	\end{equation}
	Similarly, we obtain the following by triangle inequality.
	\begin{equation}
		\d(p, o_i) \ge \d(p, p^*_i) - \d(p^*_i, o_i) \ge \d(p, p^*_i) - \tilde{r}_i \ge (1-\tfrac{\epsilon}{\c}) \cdot \d(p, p^*_i) \label{eqn:optimalcenter}
	\end{equation}
	Then, by combining equations (\ref{eqn:currentcenter}) and (\ref{eqn:optimalcenter}), we obtain the following.
	\begin{align*}
		\d(p, x_i) \le \frac{(1+\frac{\epsilon}{4})}{(1-\frac{\epsilon}{6})} \cdot \d(p, o_i) \le 
        % (1+\tfrac{\epsilon}{12}) \cdot (1+\tfrac{\epsilon}{6}) \cdot \d(p, o_i) \le 
        (1+\epsilon) \cdot \d(p, o_i),
	\end{align*}
	contradicting the fact that $(p, i) \in W$. This shows that, for each pair $(p, i) \in W$, properties $(c)$ and $(d)$ hold. Since $W$ is non-empty, this completes the proof of the lemma. 
\end{proof}

Fix some $(p_t,i_t) \in Y \times [k]$ satisfying~\Cref{lem:witnesslemma} and note that $p_t \in B_{i_t} \setminus \ball(x_{i_t},\tilde{r}_{i_t})$. Hence, consider the corresponding iteration of the \textbf{for} loop to $(p_t,i_t)$ in~\Cref{algo:apx:for}, and let $t' \in \cT$ be the corresponding child of $t$ in this iteration. We claim that $t'$ is consistent with $O$.
Towards this, let $(Q_1,\dots,Q_k)$ and  $(Q'_1,\dots,Q'_k)$  be  first arguments at $t$ and $t'$, respectively. Then, note that for $i \neq i_t$, we have, $Q'_i = Q_i$, and $Q'_{i_t}= Q_{i_t} \cup \{(p_t,\alpha_t)\}$, where $\alpha_t=\frac{\d(p_t,x_{i_t})}{1+\epsilon}$.
Hence, for any $i \neq i_t$, we have that for $(p,\alpha) \in Q'_i=Q_i$, it holds that $\d(p,o_i) \le \alpha$ due to the induction hypothesis at $t$. Similarly, for $i_t$, we have for any $(p,\alpha) \in Q'_{i_t} \setminus \{(p_{t},\alpha_t)\}$, it holds that $\d(p,o_i) \le \alpha$ due to the induction hypothesis at $t$. Now, consider $(p_t,\alpha_t) \in Q'_{i_t}$, and note that $\alpha_{t} =\frac{\d(p_t,x_{i_t})}{1+\epsilon} > \d(p_t,o_{i_t})$ since $\d(p_t, x_{i_t}) > (1+\epsilon) \cdot \d(p_t, o_{i_t})$ from~\Cref{lem:witnesslemma}, finishing the proof of the claim.
Finally, note that since none of the nodes of path $\cP$ return $\bot$, the root node $t_0$ 
% and the corresponding iteration of the \textbf{for} loop of~\Cref{algo:main} 
never returns $\bot$, finishing the proof of the lemma.

\end{proof}
Now, we are ready to finish the proof of~\Cref{thm:main}.

\noindent \paragraph*{Proof of~\Cref{thm:main}.} 
Suppose~\Cref{algo:main} terminates with a solution $(S,f_S)$. Then, we have that  $\wcost(Y,S) \le (1+5\epsilon)\cdot \cG  \le (1+5\epsilon)(1+\epsilon)\opt$. Therefore, we have
\begin{align*}
    \wcost(P,S,f_S) &\le (1+\epsilon)\cost(P,S) &\text{due to \findassgn}\\
    &\le (1+\epsilon) (1+2\epsilon)\cost(Y,S) & \text{due to coreset}\\
     &\le (1+\epsilon)^2 (1+2\epsilon)(1+5\epsilon) \opt  \\
    % &\le (1+5\epsilon)(1+\epsilon)\opt\\
        &\le (1+30\epsilon)\opt,
\end{align*}
as desired.

Now, \Cref{lem:main time} says that~\Cref{algo:main} terminates within time $\epastime$, and~\Cref{lem:correctness} implies that there is an iteration of \textbf{for} loop in~\Cref{algo:main:for} of~\Cref{algo:main} that does not return $\bot$. Hence, \Cref{algo:main} terminates with a desired solution in time $\epastime$. 
% This means that the solution returned is $(1+\epsilon)$-approximate solution to $\cI$. 
% This means~\Cref{algo:main} returns a solution $$

$\hfill\square$
% !TeX root = main.tex

\section{Applications of the Framework} \label{sec:applications}

In this section, we discuss some applications of the framework.
Key constructions in these applications are designing \genballint\ and \findassgn\ subroutines that satisfy the requirements of our framework (\Cref{subsec:formalsetup}). We divide our applications based on constraints on center selection.

\subsection{Constrained Center Selection}\label{ss:app:con}
In these problems, there are hard constraints on choosing centers as a solution. This forces designing \genballint\ (and \ballint) algorithm slightly challenging compared to vanilla clustering. For example,in the framework of~\cite{AbbasiClustering23}, it is absolutely fine for the ball intersection algorithm to chose the same center for different request sets during the execution, since the centers do not have any constraints. However, this is not allowed in problems such as capacitated clustering and fault tolerant clustering, and hence need to be addressed, depending on the problems.

\paragraph*{Preprocessing.} Towards this, we first transform the input instance $\cI$ in metric space $M=(P,F,d)$ to another instance $\cI'$ in metric space $M'=(P,F',d)$ such that the facility set partitioned into $k$ sets, i.e., $F' = \dot{\bigcup}_{t \in [k]} F'_t$, and ask for a solution  $S$ to $\cI$ to satisfy $|S \cap F_t| =1$, for all $t \in [k]$, in addition to the existing constraints of $\cI$. Finally, we want that the optimal cost of $\cI'$ is equal to that of $\cI$. This can be achieved using \emph{color coding} technique, as follows. For each $c \in F$, we assign a color $i \in [k]$ chosen uniformly at random. A \emph{good coloring} corresponds to a coloring in which all centers of a fixed optimal solution $O \subseteq F$ of size $k$ receive distinct colors. Note that this happens with probability $1/e^k$, and this procedure can be derandomized in time $2^{O(k)} n\log n$ by using standard tools from parameterized complexity~\cite{CyganFKLMPPS15}. Thus, now we assume that by inuring an additive factor of $2^{O(k)} n\log n$ and a multiplicative factor of  $2^{O(k)} \log n$ in time, the input instance has facility set partitioned into $k$ parts such that $O$ is colorful, and we want to find a \emph{colorful} set of $k$ centers that is feasible for the given constraints.
% a solution must chose one facility from each part.

\subsubsection{Capacitated Clustering} \label{subsec:capacitated}

For capacitated $k$-median, a coreset of size $\lr{\frac{k \log n}{\epsilon}}^{O(1)}$ is known~\cite{Cohen-AddadL19Capacitated}. Note that this coreset is for general metric spaces, and even works with non-uniform capacities. Further, it is possible to generalize these coresets for $(k, z)$-clustering for any fixed $z \ge 1$, where the size of the coreset  depends on $z$.

The \texttt{Assignment} subroutine in this case is simply the minimum-cost maximum flow problem, which is solvable in polynomial time, and produces integral flows.

The \genballint\ algorithm is composed of $k$ ball intersection algorithms $\{\ballint_i\}_{i \in [k]}$, one  for every request set in $\{Q_i\}_{i \in [k]}$. $\ballint_i, i \in [k]$ on request set $Q_i$, selects a center in $F_i$  with maximum capacity that satisfies all the requests in $Q_i$. This ensures that, assuming all the guesses were correct, for each $i \in [k]$, at the end of the algorithm, we will have selected a center $x_i$ such that (i) for each $p \in Y$, if $f^*$ assigns $p$ to $i$, then $\d(p, x_i) \le (1+\epsilon) \cdot \d(p, o_i)$, and (ii) capacity of $x_i$ is at least that of $o_i$, and hence  \genballint\ is compatible with \findassgn. This guarantees that there is a feasible assignment from $Y$ to $X$ of cost at most $(1+O(\epsilon))$ times $\opt$. 
% Furthermore, note that \genballint\ is compatible with \findassgn.

\textbf{A note on uniform capacities in continuous $\mathbb{R}^d$.}
When capacities are uniform, then the problem can also be defined over continuous Euclidean space, where a facility can be opened at any point in $\mathbb{R}^d$. Now, it turns out that the algorithmic scatter dimension of continuous Euclidean space is $O(\tfrac{1}{\epsilon^4}\log (\tfrac{1}{\epsilon}))$, whereas the non-algorithmic scatter dimension is $1/\epsilon^{O(d)}$. However, our framework can be adapted to work with algorithmic scatter dimension to yield EPAS for high dimensional continuous Euclidean space as follows. The \genballint\ subroutine uses the underlying \ballint\ algorithm admitted by the continuous Euclidean space for every request set $Q_i, i \in [k]$. 
However, it could happen that \genballint\ selects multiple copies of the same center in its solution $X$, and hence the assignment subroutine (minimum cost max flow subroutine) might fail. To counter this issue, we can adapt the minimum cost max flow subroutine by treating the duplicate centers in $X$ as distinct copies to obtain a feasible assignment, if there exists one. Now, since the algorithmic scatter dimension of this space is  $O(\tfrac{1}{\epsilon^4}\log (\tfrac{1}{\epsilon}))$, the depth of the recursion of our algorithm is also  $O(\tfrac{1}{\epsilon^4}\log (\tfrac{1}{\epsilon}))$, yielding EPAS running time. However, as mentioned before, it could happen that the final solution $X$ returned by our algorithm has multiple copies of the same facility. To handle this, we replace the duplicates with arbitrarily close facilities to the original facility. Note that, such facilities exist since we are in continuous Euclidean space, and it  incurs only negligible loss in the approximation factor.

\subsubsection{Matroid Clustering} \label{subsec:matroid}
In this variant of \probname, we are additionally given a matroid $\mathbb{M} = (F, \mathcal{I})$, where $\mathcal{I}$ is a collection of independent sets satisfying the matroid axioms\footnote{(1) $\emptyset \in \mathcal{I}$, (2) If $A \in \mathcal{I}$, then for every $B \subseteq A$, it holds that $B \in \mathcal{I}$, and (3) For every $A, B \in \mathcal{I}$ with $|A| < |B|$, there exists some $b \in B \setminus A$ such that $A \cup \LR{b} \in \mathcal{I}$}. The goal is to find a subset $C \subseteq F$ of centers that minimizes the $(k, z)$-clustering cost, subject to the constraint that $C$ forms an independent set in the matroid, i.e., $C \in \mathcal{I}$. We assume that the matroid $\mathcal{M}$ is provided via an independence oracle, i.e., given a subset $S \subseteq F$, the oracle reports whether $S \in \mathcal{I}$ in a single step. We assume that the rank of $\mathcal{M}$ is $k$, i.e., we are looking for a subset of facilities that is a base in the matroid $\mathcal{M}$ -- this is without loss of generality, since we can always truncate the matroid to rank $k$ if necessary (the independence oracle of the truncated matroid can be easily simulated given the oracle for $\mathcal{M}$). As mentioned earlier, using  color coding, we partition $F$ into $k$ colors such that $O$ is colorful. This is nothing but  creating a partition matroid $\mathbb{P}$ on $F$. Overall, this means, we are looking for a set of centers in the intersection of $\mathbb{M}$ and $\mathbb{P}$.

The \genballint\ algorithm is simply the polynomial time algorithm for the classic problem of matroid intersection~\cite{edmonds2003submodular} of $\mathbb{M}$ and $\mathbb{P}$.

The \findassgn\ simply assigns a point completely to the closest center in $X$. In this case, both \genballint\ and \findassgn\ are compatible with each other.

\subsubsection{Fault-tolerant Clustering} \label{subsec:fault}
We consider (uniform) fault-tolerant setting, which is defined as follows.
\begin{definition} \label{def:fkm}
Let $z \ge 1$ be a fixed constant.
    The input consists of a metric space $\mathcal{M} = (P \cup F, d)$, and positive integers $k\ge 1$ and $\ell \le k$. The goal is to find $k$ centers $X \subseteq F$ to minimize $\cost(X) := \sum_{p \in P} \d_\ell(p,X)$, where, for a point $p \in P$, $\d_\ell(p,X) := \min_{X\ \subseteq X, |X'|=\ell} \sum_{x'\in X'} (\d(p,x'))^z$, is   the sum of cost of serving $p$ by $\ell$ nearest centers in $X$.
\end{definition}

Now we show how fault-tolerant clustering is capture by our framework. Let $X$ be a set of $k$ centers, and let $X_p$ be the set of $\ell$ closest centers in $X$ to $p \in P$. Then, we have that $d_\ell(p,X) = \sum_{x \in X_p} d(p, x)$, and $\cost(X) = \sum_{p \in P} \sum_{x \in X_p} d(p, x)$. Let $f:P \times X \rightarrow \mathbb{R}_{\ge 0}$ be the corresponding assignment function, i.e., $f(p,x) =1$ if  $x \in X_p$, and $0$ otherwise.
In our setting, we will instead work with normalized point-cost $d'_\ell(p,X)$, which is defined as $d'_\ell(p,X) := \sum_{x \in X_p} \tfrac{d(p, x)}{\ell}$, and the normalized total cost $\cost'(X):=\sum_{p \in P} \d'_\ell(p,X)$. Let $f':P \times X \rightarrow \mathbb{R}_{\ge 0}$ be the corresponding assignment function, i.e.,  $f(p,x) =1/\ell$ if  $x \in X_p$, and $0$ otherwise. Then, note that $\sum_{x \in X} f'(p,x)=1$, for every $p \in P$, as required. Moreover, 
\[\cost(X) = \sum_{p \in P} \sum_{x \in X} f(p,x) (d(p,x))^z = \ell \cdot\big( \sum_{p \in P} \sum_{x \in X} f'(p,x) (d(p,x))^z\big) = \ell \cdot \cost'(X).\]

Thus, for any center set $X$, the cost of $f'$ is exactly $\ell$ times that of $f'$. Moreover, $f'$ satisfies the condition of our problem. Hence, in our framework, we instead work with $f'$ and the estimate of cost of $(O,f')$. 

Fault-tolerant $(k,z)$-Clustering is known to admit coreset of size $O(m+(\log n + k +\epsilon^{-1})k^2 \epsilon^{-2z})$ in general metrics due to~\cite{Huang0L025}. As mentioned earlier, we assume that the facilities are colored in $k$ colors such that $O$ is colorful.

% \genballint\ subroutine: 
For every request set $Q_i, i \in [k]$, the \genballint\ subroutine returns a feasible facility of color $i$.
% Let $\ballint$ be the ball intersection algorithm admitted by $\cM$.
% For every request set $Q_i, i \in [k]$, the \genballint\ subroutine invokes $\ballint$ for metric space $M_i=(P,F_i,d)$ on $Q_i$. Recall that $F_i \subseteq F$ is a set of color $i$ facilities, and moreover, it holds that $M_i \in \cM$. 

Given a set $X$ of $k$ centers, the \texttt{Assignment} subroutine in this case is the following:  for a point $p \in Y$ with weight $w(p)$, let $X_p \subseteq X$ be the set of $\ell$ closest centers in $X$ to $p$. Then, assign weight $w(p)/\ell$ to each center in $X_p$. In this case, both \genballint\ and \findassgn\ are compatible with each other.
% $\ell$ closest centers in $X$ with weight $w(p)/\ell$.

\subsection{Unconstrained Center Selection} \label{ss:app:un}
In this section, we describe some applications of our framework which does not enforce restrictions on the center selection, however, there may be restrictions on how points are assigned to the chosen centers. Therefore, in these applications, our framework yields EPAS for metrics for bounded algorithm $\epsilon$-scatter dimension, since we can use the best \ballint\ admitted by the metric in our framework.

\subsubsection{\texorpdfstring{$(\alpha, \beta)$}{(alpha, beta)}-Fair Clustering and Variants} \label{subsec:fairclustering}

We begin with the definition of $(\alpha, \beta)$-fair $k$-median problem that has been studied in \cite{CKLV17,BravermanCJKST022,BandyapadhyayFS24,Huang0L025}. 

\begin{definition} \label{def:fkm}
    The input consists of a metric space $\mathcal{M} = (P \cup F, d)$, where $P = P_1 \cup P_2 \cup \ldots P_\ell$, where $P_i$'s are (not necessarily disjoint) \emph{groups} of clients in $P$. Further, the input also consists of two fairness vectors $\alpha, \beta \in [0, 1]^\ell$. The objective is to select a set $X \subseteq F$ of $k$ centers, and an assignment $f: P \to X$ such that $f$ satisfies the following fairness constraints:   \begin{align*}
        \alpha_i \le \frac{|\LR{x \in P_i: f(x) = c}|}{|\LR{x \in P: f(x) = c}|} \le \beta_i \qquad \forall c \in X,  \forall i \in [\ell]
    \end{align*}
    And further, $\cost(P, X) \coloneqq \sum_{p \in P} \d(p, f(p))$ is minimized among all such assignments. In this problem, $\Gamma$ denotes the number of \emph{equivalence classes} of points w.r.t.~groups, i.e., the number of distinct subsets of groups that a point can belong.
\end{definition}
The following definitions pertaining to coresets for $(\alpha, \beta)$-fair clustering were introduced in \cite{BandyapadhyayFS24}.
\begin{definition}[Coloring Constraint and Universal Coresets, \cite{BandyapadhyayFS24}] \label{def:universalcoresets}
    Fix an instance $(\mathcal{M} = (P \cup F, \d), (P_1, P_2, \ldots, P_\ell), \alpha, \beta, k)$ of the $(\alpha, \beta)$-fair $k$-median problem.
    \begin{itemize}
        \item A \emph{coloring constraint} is a $k \times \ell$ matrix $M$ containing non-negative entries. 
        \item Let $(Y, w)$ be a weighted set of points. Then, for a set of $k$ centers $X = \LR{c_1, c_2, \ldots, c_k}$ and a coloring constraint $M \in \mathbb{N}^{k \times \ell}$, $\wcost(Y, X, M)$ is defined as the minimum value of $\sum_{p \in Y} \sum_{i \in [k]} f(p, c_i) \cdot \d(p, c_i)$, over all assignments $f: Y \times X \to \mathbb{N}$ satisfying the following constraints. 
        \begin{itemize}
            \item For each $p \in P$, $\sum_{i \in [k]} f(p, c_i) = w(p)$, and
            \item For each $i \in [k]$ and each group $j \in [\ell]$, $\sum_{p \in P_j} f(p, c_i) = M_{ij}$.
        \end{itemize}
        If no such assignment exists, then $\wcost(Y, X, M)$ is defined as $\infty$. Further, when $(Y, w)$ is the original set of points $P$ with unit weight function $\mathbbm{1}$, then $\wcost(P, X, M)$ is written as $\cost(P, X, M)$.
        \item A \emph{universal coreset} is a pair $(Y, w)$, such that for every set $X$ of $k$ centers and every coloring constraint $M$, the following holds:
        $$(1-\epsilon) \cdot \cost(P, X, M) \le \wcost(Y, X, M) \le (1+\epsilon) \cdot \cost(P, X, M).$$
    \end{itemize}
\end{definition}
Universal coresets of size $\Gamma \lr{\frac{k \log n}{\epsilon}}^{O(1)}$ in general metrics were designed in \cite{BandyapadhyayFS24}, and the extensions to $(k, z)$-clustering variant (which is defined by naturally modifying the definitions above by including the $z$-th power of distances) with improved size were given in \cite{Huang0L025}. For our purpose, the former guarantee on the size suffices. Further, \cite{BandyapadhyayFS24} also designed an algorithm that takes input a weighted point set $(Y, w)$, and a set $X = \LR{c_1, c_2, \ldots, c_k}$ of $k$ centers, runs in time $(k\Gamma)^{O(k \Gamma)} \cdot n^{O(1)}$, and returns an an assignment $f: Y \times X \to \mathbb{N}$, such that the cost $\sum_{p \in Y} \sum_{i \in [k]} f(p, c_i) \cdot \d(p, c_i)$ is minimized across all assignments that respect the $\alpha, \beta$ fairness constraints as above. Further, it is easy to check that this algorithm does not depend on the objective function -- in fact it can compute an optimal assignment for any $(k, z)$-clustering objective.
For \genballint, we use the underlying \ballint\ algorithm to compute a feasible center for every request set $Q_i, i \in [k]$.
Thus, we have all the ingredients to plug into our framework for $(\alpha, \beta)$-fair $(k, z)$-clustering in metric spaces of bounded scatter dimension for any fixed $z \ge 1$, and obtain an EPAS with running time $\Gamma^{O(\Gamma k \lambda(\epsilon) \log(\frac{1}{\epsilon}) (\log k + \log(\lambda(\epsilon))))} \cdot (|\cI|)^{O(1)}$.

\paragraph*{Further implications.} \cite{BandyapadhyayFS24} discuss the following variants of constrained clustering in their paper.
\begin{description}
    \item[Lower Bounded Clustering.] Here, the points have only a single color, but the size of each cluster must be at least $L$, where $L$ is a non-negative integer given in the input.
    \item[$\ell$-Diversity Clustering.] Here, the points are again classified in one of $\ell$ groups, $P_1, P_2, \ldots, P_\ell$ (a point can belong to multiple groups), and each cluster $A$ should satisfy that $|A \cap P_\ell| \le \frac{|A|}{\ell}$. 
    \item[Chromatic Clustering.] The setting is similar to above, but each cluster can contain at most one point of each group.
\end{description}
\cite{BandyapadhyayFS24} discuss in detail the modifications (if any) to the universal coreset and the FPT algorithm for finding an optimal assignment required to handle each of these variants, and we direct the reader to these discussions. Since these are the only two ingredients required in our problem, our framework implies an EPAS for all of these variants in metric spaces of bounded scatter dimension.

\subsubsection{Vanilla (Voronoi) Clustering} \label{subsec:vanilla}

Our framework also naturally for vanilla clustering, implying an EPAS with running time $2^{O_\epsilon(k \log k)} \cdot (|\cI|)^{O(1)}$. Notice that this already improves upon the $2^{O_\epsilon(k \log^2 k)} \cdot n^{O(1)}$ running time from \cite{AbbasiClustering23}, where $O_\epsilon(\cdot)$ notation hides the dependence on $\epsilon$. However, the dependence on $\epsilon$ may be worse for metric spaces with super-exponential scattter dimension. Here, we describe a simple modification for choosing a point of the coreset in the algorithm (as opposed to the for loop of line 4 in \Cref{algo:apx}) that, (i) has improved dependence on $k$ as described above, and (ii) matches the dependence on $\epsilon$ as in \cite{AbbasiClustering23}, thus improving the running time as compared to \cite{AbbasiClustering23} in all regimes.

To this end, first we make the following observation: in Voronoi clustering, without loss of generality, it can be assumed that for each $p \in Y$, the entire weight $w(p)$ is assigned to a single cluster-center $x_i \in X$ that is closest to $p$. 
Then, in \Cref{algo:apx} instead of the for loop of lines 4 to 7: we sample a point $p \in \bigcup_{i \in [k]} B_i$ according to the following distribution: $\Pr(p = a) = \frac{w(a) \d(a, X)}{\sum_{p \in \cup_i B_i} w(p) \d(p, X) }$ for $a \in \cup_{i \in [k]} B_i$. Let $p$ be the sampled point. Next, we also sample an index $i$ uniformly at random from the set $\LR{j \in [k]: p \in B_i}$. Then, we make a similar recursive call to the same algorithm by adding the request $(p, \frac{\d(p, x_i)}{1+\epsilon/3})$ to $Q_i$, similar to line 5. The crucial observation is the following, which is inspired from similar claims from \cite{AbbasiClustering23,DBLP:conf/stacs/Gadekar025}.

\begin{claim}
    Let $X$ be a solution such that $\cost(Y, X) > (1+5\epsilon) \cdot \mathcal{G}$, and let $W \coloneqq \{p \in Y: \d(p, X) > (1+\epsilon) \cdot \d(p, O)\}$. Then, (i) $W \subseteq \bigcup_{i \in [k]} B_i$, and (ii) $\frac{\sum_{p \in W} w(p) \d(p, X)}{\sum_{p \in Y} w(p) \d(p, X)} \ge \frac{\epsilon}{10}$.
\end{claim}
\begin{proof}
    Consider  $X$ and the corresponding set $W$ as in the statement. Let $H \coloneqq Y \setminus W$, and note that for all points $p \in H$, $\d(p, X) \le (1+\epsilon) \cdot \d(p, O)$. For any $T \subseteq Y$, let $C_T \coloneqq \sum_{p \in T} w(p) \d(p, X)$. 

    First, let us suppose that there exists a point $p \in W$ such that $\d(p, x_i) > \frac{6r'_i}{\epsilon}$ for all $i \in [k]$. Then, by arguments similar to that in \Cref{lem:witnesslemma}, we can argue that $\d(p, X) < (1+\epsilon) \d(p, O)$, contradicting the assumption that $p \in W$. This shows part (i).

    Let us now suppose for contradiction for part (ii) that $\frac{C_W}{C_Y} < \frac{\epsilon}{10}$.%, which implies that $C_H = C_Y - C_W \ge (1-\frac{\epsilon}{10}) C_Y$. 
    Then, consider:
    \begin{align*}
        C_H = \sum_{p \in H} w(p) \d(p, X) &\le (1+\epsilon) \cdot \sum_{p \in H} w(p) \d(p, O) \le (1+\epsilon) \opt
    \end{align*}
    Now, note that $C_Y = C_H + C_W \le (1+\epsilon) \opt + \frac{\epsilon}{10} C_Y$, which implies that $(1-\frac{\epsilon}{10}) C_Y \le (1+\epsilon) \opt \implies C_Y \le \frac{1+\epsilon}{1-\epsilon/10} \opt \le (1+2\epsilon) \opt \le (1+2\epsilon) \mathcal{G}$, contradicting the hypothesis. 
\end{proof}
Assuming this claim, we know that the probability that the point $p$ sampled according to the aforementioned distribution belongs to $W$ is at least $\frac{\epsilon}{10}$. Conditioned on this event, with probability at least $\frac{1}{k}$, the sampled cluster index $i$ in the aforementioned way is same as that of $p$ in the optimal solution $O$. Hence, the state of the algorithm remains consistent with $O$. The depth of the recursion remains bounded by $O(k \log(\frac{1}{\epsilon}) \cdot \lambda(\epsilon))$, which implies that we obtain an algorithm with running time $2^{O(k \log(\frac{1}{\epsilon}) \lambda(\epsilon/c) \cdot \log(\frac{k}{\epsilon}))} \cdot (|\cI|)^{O(1)}$, for some constant $c$ that depends only on $z$. 

%, where we obtain a faster algorithm as compared to \cite{AbbasiClustering23} for $(k, z)$-clustering for any fixed $z \ge 1$ \red{with a few key modifications to the algorithm and analysis}. Indeed, coresets of size $\frac{k \log n}{\epsilon}^{O(1)}$ are known for $(k, z)$-clustering for any fixed $z \ge 1$ in general metric spaces. Here, the \texttt{Assignment} algorithm is simply the voronoi assignment, i.e., each point $p \in Y$ is assigned to its closest center in $X$. \todo{modifications, if required}.

\subsection{Additional Remarks} \label{subsec:remarks}

\paragraph*{Handling Outliers.} Jaiswal and Kumar~\cite{Jaiswal023} (also \cite{DabasGI25}), extending the work of Agrawal et al.~\cite{AgrawalISX23} gave a general \emph{additive-$\epsilon$ approximation-preserving} reduction from constrained $(k, z)$-clustering with $m$ outliers to its analogous constrained version \emph{without} outliers that has a multiplicative overhead of $\lr{\frac{k+m}{\epsilon}}^{O(m)} \cdot n^{O(1)}$. By plugging in our EPASes for constrained clustering into this result, we also obtain EPASes for the corresponding outlier versions. 

\paragraph*{Relations to Improved Coresets for Constrained Clustering.} For most known metric spaces with bounded scatter dimension, the recent work of Huang et al.~\cite{Huang0L025} gives coresets of size $\text{poly}(k, \epsilon)$ for a large class of constrained $(k, z)$-clustering problems, including capacitated clustering, $(\alpha, \beta)$-fair clustering, fault-tolerant clustering, and more. To put this work into a proper context w.r.t.~our EPASes, two remarks are in order.

Firstly, we note that such coresets can be used to obtain EPASes for the corresponding problems\footnote{Although the following argument is very simple and should be folklore, we are unable to find a reference that explicitly describes how coresets of size independent of $\log n$ may be used to obtain EPASes for \emph{constrained} $(k, z)$-clustering. Therefore, we provide the argument for the sake of completeness in order to perform a fair comparison with our results.}. 
Let $(Y, w)$ be the coreset of size $\text{poly}(k, \epsilon)$. 
We start, as usual, with the color coding step, incurring a $O(2^{O(k)} |\cI|^{O(1)})$ factor. Then, for each $i \in [k]$, we guess the cost of the $i$th optimal cluster, up to a factor of $(1+\varepsilon)$; call this value $\tilde{C}_i$.
Next, for each $p \in Y$ and each $i \in [k]$, we guess (up to a factor of $(1+\varepsilon)$) the portion of weight of $p$ assigned to the $i$th cluster, call this $f_i(p)$. Note that the number of guesses is bounded by $(\frac{\log n}{\epsilon})^{k \cdot |Y|}$, which is FPT in $k$ and $\varepsilon$. Fix one such guess. Then, for each cluster $i \in [k]$, let 
$\hat{F}_i$ be the subset of centers $c$ of color $i$ such that $\sum_{p \in Y} \d(p, c) \cdot f_i(p) \approx_{1+\varepsilon} \tilde{C}_i$. 
Then, we select a ``best center'' from $\hat{F}_i$ depending upon the underling constraint (for example, for capacitated setting, we select the largest capacity center in $\hat{F}_i$).
This yields us a candidate center set $X$. Finally, compute an optimal assignment using the assignment algorithm in time FPT in $k$ and $\epsilon$. Finally, we return the minimum-cost solution found over all guesses. Although this strategy also yields an EPAS for many of the constrained problems in metric spaces of bounded scatter dimension, we note that the running time is $(k|Y|/\epsilon)^{O(k |Y|)} \cdot n^{O(1)}$, and the size of coresets is $\Omega_{\epsilon}(k^2)$. Therefore, the dependence on $k$ of the running time is $2^{\Omega_\epsilon(k^3 \log k)} \cdot n^{O(1)}$, which is much worse compared to our running time of $2^{O_\epsilon(k \log k)} \cdot n^{O(1)}$. Secondly, although the result of \cite{Huang0L025} gives coresets of size independent of $\log n$ in most known metric spaces of bounded scatter dimension, it is \emph{not known} to yield such coresets for all metric spaces of bounded scatter dimension. Thus, this result does not answer the open question of \cite{AbbasiClustering23} in its entirety.

Secondly, we note that plugging in the improved coresets into our algorithm improves the dependence on $\epsilon$ in running time of the algorithm -- a factor of $O(\log(\lambda(\epsilon)))$ in the exponent is replaced by a factor of $O(\log(\frac{1}{\epsilon}))$. Since the scatter dimension $\lambda(\epsilon
)$ of all of these classes of metric spaces is bounded by $2^{1/\epsilon^{O(1)}}$, this results in an improved dependence on $\epsilon$.

\bibliographystyle{alpha}
\bibliography{references}

@article{Har-PeledM05,
	author = {Sariel Har{-}Peled and Soham Mazumdar},
	bibsource = {dblp computer science bibliography, https://dblp.org},
	doi = {10.1007/S00453-004-1123-0},
	journal = {Algorithmica},
	number = {3},
	pages = {147--157},
	title = {Fast Algorithms for Computing the Smallest k-Enclosing Circle},
	url = {https://doi.org/10.1007/s00453-004-1123-0},
	volume = {41},
	year = {2005}}

@inproceedings{edmonds2003submodular,
  title={Submodular functions, matroids, and certain polyhedra},
  author={Edmonds, Jack},
  booktitle={Combinatorial Optimization—Eureka, You Shrink! Papers Dedicated to Jack Edmonds 5th International Workshop Aussois, France, March 5--9, 2001 Revised Papers},
  pages={11--26},
  year={2003},
  organization={Springer}
}

@inproceedings{CharikarKMN01,
	author = {Moses Charikar and Samir Khuller and David M. Mount and Giri Narasimhan},
	bibsource = {dblp computer science bibliography, https://dblp.org},
	booktitle = {Proceedings of the Twelfth Annual Symposium on Discrete Algorithms (SODA)},
	pages = {642--651},
	publisher = {{ACM/SIAM}},
	title = {Algorithms for facility location problems with outliers},
	url = {http://dl.acm.org/citation.cfm?id=365411.365555},
	year = {2001}}

@article{feldman2013turning,
	author = {Feldman, Dan and Schmidt, Melanie and Sohler, Christian},
	journal = {SIAM Journal on Computing},
	number = {3},
	pages = {601--657},
	publisher = {SIAM},
	title = {Turning big data into tiny data: Constant-size coresets for k-means, PCA, and projective clustering},
	volume = {49},
	year = {2020}}

@inproceedings{sohler2018strong,
	author = {Sohler, Christian and Woodruff, David P},
	booktitle = {2018 IEEE 59th Annual Symposium on Foundations of Computer Science (FOCS)},
	organization = {IEEE},
	pages = {802--813},
	title = {Strong coresets for k-median and subspace approximation: Goodbye dimension},
	year = {2018}}

@inproceedings{feldman2011unified,
	author = {Feldman, Dan and Langberg, Michael},
	booktitle = {Proceedings of the forty-third annual ACM symposium on Theory of computing (STOC)},
	pages = {569--578},
	title = {A unified framework for approximating and clustering data},
	year = {2011}}

@article{cohen2021near,
	author = {Cohen-Addad, Vincent and Feldmann, Andreas Emil and Saulpic, David},
	journal = {Journal of the ACM},
	number = {6},
	pages = {1--34},
	publisher = {ACM New York, NY},
	title = {Near-linear time approximation schemes for clustering in doubling metrics},
	volume = {68},
	year = {2021}}

@inproceedings{FeldmannV22,
  author       = {Andreas Emil Feldmann and
                  Tung Anh Vu},
  title        = {Generalized k-Center: Distinguishing Doubling and Highway Dimension},
  booktitle    = {Graph-Theoretic Concepts in Computer Science - 48th International
                  Workshop, {WG} 2022, T{\"{u}}bingen, Germany, June 22-24, 2022,
                  Revised Selected Papers},
  series       = {Lecture Notes in Computer Science},
  volume       = {13453},
  pages        = {215--229},
  publisher    = {Springer},
  year         = {2022},
  url          = {https://doi.org/10.1007/978-3-031-15914-5\_16},
  doi          = {10.1007/978-3-031-15914-5\_16},
  bibsource    = {dblp computer science bibliography, https://dblp.org}
}

@article{KumarSS10,
	author = {Amit Kumar and Yogish Sabharwal and Sandeep Sen},
	journal = {J. {ACM}},
	number = {2},
	pages = {5:1--5:32},
	title = {Linear-time approximation schemes for clustering problems in any dimensions},
	volume = {57},
	year = {2010}}

@inproceedings{DBLP:conf/stoc/Har-PeledM04,
	author = {Sariel Har{-}Peled and Soham Mazumdar},
	booktitle = {Proceedings of the 36th Annual {ACM} Symposium on Theory of Computing, Chicago, IL, USA, June 13-16, 2004},
	pages = {291--300},
	publisher = {{ACM}},
	title = {On coresets for k-means and k-median clustering},
	year = {2004}}

@inproceedings{DBLP:conf/stoc/HuangV20,
	author = {Lingxiao Huang and Nisheeth K. Vishnoi},
	bibsource = {dblp computer science bibliography, https://dblp.org},
	booktitle = {Proccedings of the 52nd Annual {ACM} {SIGACT} Symposium on Theory of Computing, {STOC} 2020, Chicago, IL, USA, June 22-26, 2020},
	doi = {10.1145/3357713.3384296},
	pages = {1416--1429},
	publisher = {{ACM}},
	title = {Coresets for clustering in Euclidean spaces: importance sampling is nearly optimal},
	url = {https://doi.org/10.1145/3357713.3384296},
	year = {2020}}

@inproceedings{Cohen-AddadG0LL19Tight,
	author = {Vincent Cohen{-}Addad and Anupam Gupta and Amit Kumar and Euiwoong Lee and Jason Li},
	bibsource = {dblp computer science bibliography, https://dblp.org},
	booktitle = {46th International Colloquium on Automata, Languages, and Programming (ICALP)},
	doi = {10.4230/LIPIcs.ICALP.2019.42},
	pages = {42:1--42:14},
	publisher = {Schloss Dagstuhl - Leibniz-Zentrum f{\"{u}}r Informatik},
	series = {LIPIcs},
	title = {Tight {FPT} Approximations for $k$-Median and $k$-Means},
	url = {https://doi.org/10.4230/LIPIcs.ICALP.2019.42},
	volume = {132},
	year = {2019}}

@inproceedings{braverman2021coresets,
	author = {Braverman, Vladimir and Jiang, Shaofeng H-C and Krauthgamer, Robert and Wu, Xuan},
	booktitle = {Proceedings of the 2021 ACM-SIAM Symposium on Discrete Algorithms (SODA)},
	organization = {SIAM},
	pages = {2679--2696},
	title = {Coresets for clustering in excluded-minor graphs and beyond},
	year = {2021}}

@inproceedings{DBLP:conf/icml/BakerBHJK020,
	author = {Daniel N. Baker and Vladimir Braverman and Lingxiao Huang and Shaofeng H.{-}C. Jiang and Robert Krauthgamer and Xuan Wu},
	bibsource = {dblp computer science bibliography, https://dblp.org},
	booktitle = {Proceedings of the 37th International Conference on Machine Learning, {ICML} 2020, 13-18 July 2020, Virtual Event},
	pages = {569--579},
	publisher = {{PMLR}},
	series = {Proceedings of Machine Learning Research},
	title = {Coresets for Clustering in Graphs of Bounded Treewidth},
	url = {http://proceedings.mlr.press/v119/baker20a.html},
	volume = {119},
	year = {2020}}

@book{CyganFKLMPPS15,
	author = {Marek Cygan and Fedor V. Fomin and Lukasz Kowalik and Daniel Lokshtanov and D{\'{a}}niel Marx and Marcin Pilipczuk and Michal Pilipczuk and Saket Saurabh},
	bibsource = {dblp computer science bibliography, https://dblp.org},
	doi = {10.1007/978-3-319-21275-3},
	isbn = {978-3-319-21274-6},
	publisher = {Springer},
	title = {Parameterized Algorithms},
	url = {https://doi.org/10.1007/978-3-319-21275-3},
	year = {2015}}

@inproceedings{AbbasiClustering23,
	author = {Fateme Abbasi and Sandip Banerjee and Jaroslaw Byrka and Parinya Chalermsook and Ameet Gadekar and Kamyar Khodamoradi and D{\'{a}}niel Marx and Roohani Sharma and Joachim Spoerhase},
	bibsource = {dblp computer science bibliography, https://dblp.org},
	booktitle = {64th {IEEE} Annual Symposium on Foundations of Computer Science, {FOCS} 2023, Santa Cruz, CA, USA, November 6-9, 2023},
	doi = {10.1109/FOCS57990.2023.00085},
	pages = {1377--1399},
	publisher = {{IEEE}},
	title = {Parameterized Approximation Schemes for Clustering with General Norm Objectives},
	url = {https://doi.org/10.1109/FOCS57990.2023.00085},
	year = {2023}}

@inproceedings{BourneufP25,
	author = {Romain Bourneuf and Marcin Pilipczuk},
	booktitle = {Proceedings of the 2025 Annual ACM-SIAM Symposium on Discrete Algorithms (SODA) (to appear)},
	title = {Bounding $\varepsilon$-scatter dimension via metric sparsity},
	year = {2025}}

@inproceedings{Cohen-AddadL19Capacitated,
	author = {Vincent Cohen{-}Addad and Jason Li},
	bibsource = {dblp computer science bibliography, https://dblp.org},
	booktitle = {46th International Colloquium on Automata, Languages, and Programming, {ICALP} 2019, July 9-12, 2019, Patras, Greece},
	doi = {10.4230/LIPICS.ICALP.2019.41},
	pages = {41:1--41:14},
	publisher = {Schloss Dagstuhl - Leibniz-Zentrum f{\"{u}}r Informatik},
	series = {LIPIcs},
	title = {On the Fixed-Parameter Tractability of Capacitated Clustering},
	url = {https://doi.org/10.4230/LIPIcs.ICALP.2019.41},
	volume = {132},
	year = {2019}}

@article{BandyapadhyayFS24,
	author = {Sayan Bandyapadhyay and Fedor V. Fomin and Kirill Simonov},
	bibsource = {dblp computer science bibliography, https://dblp.org},
	doi = {10.1016/J.JCSS.2024.103506},
	journal = {J. Comput. Syst. Sci.},
	pages = {103506},
	title = {On coresets for fair clustering in metric and Euclidean spaces and their applications},
	url = {https://doi.org/10.1016/j.jcss.2024.103506},
	volume = {142},
	year = {2024}}

@article{AgrawalISX23,
	author = {Akanksha Agrawal and Tanmay Inamdar and Saket Saurabh and Jie Xue},
	bibsource = {dblp computer science bibliography, https://dblp.org},
	doi = {10.1613/JAIR.1.14883},
	journal = {J. Artif. Intell. Res.},
	pages = {143--166},
	title = {Clustering what Matters: Optimal Approximation for Clustering with Outliers},
	url = {https://doi.org/10.1613/jair.1.14883},
	volume = {78},
	year = {2023}}

@inproceedings{KrishnaswamyLS18,
	author = {Ravishankar Krishnaswamy and Shi Li and Sai Sandeep},
	bibsource = {dblp computer science bibliography, https://dblp.org},
	booktitle = {Proceedings of the 50th Annual {ACM} {SIGACT} Symposium on Theory of Computing, {STOC} 2018, Los Angeles, CA, USA, June 25-29, 2018},
	doi = {10.1145/3188745.3188882},
	pages = {646--659},
	publisher = {{ACM}},
	title = {Constant approximation for k-median and k-means with outliers via iterative rounding},
	url = {https://doi.org/10.1145/3188745.3188882},
	year = {2018}}

@article{DBLP:journals/tcs/KhullerPS00,
  author       = {Samir Khuller and
                  Robert Pless and
                  Yoram J. Sussmann},
  title        = {Fault tolerant K-center problems},
  journal      = {Theor. Comput. Sci.},
  volume       = {242},
  number       = {1-2},
  pages        = {237--245},
  year         = {2000},
  url          = {https://doi.org/10.1016/S0304-3975(98)00222-9},
  doi          = {10.1016/S0304-3975(98)00222-9},
  bibsource    = {dblp computer science bibliography, https://dblp.org}
}

@InProceedings{10.1007/978-3-642-13036-6_19,
author="Byrka, Jaroslaw
and Srinivasan, Aravind
and Swamy, Chaitanya",
editor="Eisenbrand, Friedrich
and Shepherd, F. Bruce",
title="Fault-Tolerant Facility Location: A Randomized Dependent LP-Rounding Algorithm",
booktitle="Integer Programming and Combinatorial Optimization",
year="2010",
publisher="Springer Berlin Heidelberg",
address="Berlin, Heidelberg",
pages="244--257",
isbn="978-3-642-13036-6"
}

@article{DBLP:journals/talg/HajiaghayiHLLS16,
  author       = {Mohammad Taghi Hajiaghayi and
                  Wei Hu and
                  Jian Li and
                  Shi Li and
                  Barna Saha},
  title        = {A Constant Factor Approximation Algorithm for Fault-Tolerant \emph{k}-Median},
  journal      = {{ACM} Trans. Algorithms},
  volume       = {12},
  number       = {3},
  pages        = {36:1--36:19},
  year         = {2016},
  url          = {https://doi.org/10.1145/2854153},
  doi          = {10.1145/2854153},
  bibsource    = {dblp computer science bibliography, https://dblp.org}
}

@inproceedings{DBLP:conf/soda/KrishnaswamyKNSS11,
  author       = {Ravishankar Krishnaswamy and
                  Amit Kumar and
                  Viswanath Nagarajan and
                  Yogish Sabharwal and
                  Barna Saha},
  editor       = {Dana Randall},
  title        = {The Matroid Median Problem},
  booktitle    = {Proceedings of the Twenty-Second Annual {ACM-SIAM} Symposium on Discrete
                  Algorithms, {SODA} 2011, San Francisco, California, USA, January 23-25,
                  2011},
  pages        = {1117--1130},
  publisher    = {{SIAM}},
  year         = {2011},
  url          = {https://doi.org/10.1137/1.9781611973082.84},
  doi          = {10.1137/1.9781611973082.84},
  bibsource    = {dblp computer science bibliography, https://dblp.org}
}

@article{DBLP:journals/talg/Swamy16,
  author       = {Chaitanya Swamy},
  title        = {Improved Approximation Algorithms for Matroid and Knapsack Median
                  Problems and Applications},
  journal      = {{ACM} Trans. Algorithms},
  volume       = {12},
  number       = {4},
  pages        = {49:1--49:22},
  year         = {2016},
  url          = {https://doi.org/10.1145/2963170},
  doi          = {10.1145/2963170},
  bibsource    = {dblp computer science bibliography, https://dblp.org}
}

@inproceedings{DBLP:conf/soda/Chen08,
  author       = {Ke Chen},
  editor       = {Shang{-}Hua Teng},
  title        = {A constant factor approximation algorithm for \emph{k}-median clustering
                  with outliers},
  booktitle    = {Proceedings of the Nineteenth Annual {ACM-SIAM} Symposium on Discrete
                  Algorithms, {SODA} 2008, San Francisco, California, USA, January 20-22,
                  2008},
  pages        = {826--835},
  publisher    = {{SIAM}},
  year         = {2008},
  url          = {http://dl.acm.org/citation.cfm?id=1347082.1347173},
  bibsource    = {dblp computer science bibliography, https://dblp.org}
}

@article{KhullerS00,
	author = {Samir Khuller and Yoram J. Sussmann},
	bibsource = {dblp computer science bibliography, https://dblp.org},
	doi = {10.1137/S0895480197329776},
	journal = {{SIAM} J. Discret. Math.},
	number = {3},
	pages = {403--418},
	title = {The Capacitated \emph{K}-Center Problem},
	url = {https://doi.org/10.1137/S0895480197329776},
	volume = {13},
	year = {2000}}

@inproceedings{CyganHK12,
	author = {Marek Cygan and MohammadTaghi Hajiaghayi and Samir Khuller},
	bibsource = {dblp computer science bibliography, https://dblp.org},
	booktitle = {53rd Annual {IEEE} Symposium on Foundations of Computer Science, {FOCS} 2012, New Brunswick, NJ, USA, October 20-23, 2012},
	doi = {10.1109/FOCS.2012.63},
	pages = {273--282},
	publisher = {{IEEE} Computer Society},
	title = {{LP} Rounding for k-Centers with Non-uniform Hard Capacities},
	url = {https://doi.org/10.1109/FOCS.2012.63},
	year = {2012}}

@inproceedings{Li16,
	author = {Shi Li},
	bibsource = {dblp computer science bibliography, https://dblp.org},
	booktitle = {Proceedings of the Twenty-Seventh Annual {ACM-SIAM} Symposium on Discrete Algorithms, {SODA} 2016, Arlington, VA, USA, January 10-12, 2016},
	doi = {10.1137/1.9781611974331.CH56},
	pages = {786--796},
	publisher = {{SIAM}},
	title = {Approximating capacitated \emph{k}-median with {(1} + {$\epsilon$})\emph{k} open facilities},
	url = {https://doi.org/10.1137/1.9781611974331.ch56},
	year = {2016}}

@article{Li17,
	author = {Shi Li},
	bibsource = {dblp computer science bibliography, https://dblp.org},
	doi = {10.1145/2983633},
	journal = {{ACM} Trans. Algorithms},
	number = {2},
	pages = {22:1--22:18},
	title = {On Uniform Capacitated \emph{k}-Median Beyond the Natural {LP} Relaxation},
	url = {https://doi.org/10.1145/2983633},
	volume = {13},
	year = {2017}}

@inproceedings{AdamczykBMM019,
	author = {Marek Adamczyk and Jaroslaw Byrka and Jan Marcinkowski and Syed Mohammad Meesum and Michal Wlodarczyk},
	bibsource = {dblp computer science bibliography, https://dblp.org},
	booktitle = {27th Annual European Symposium on Algorithms, {ESA} 2019, September 9-11, 2019, Munich/Garching, Germany},
	doi = {10.4230/LIPICS.ESA.2019.1},
	pages = {1:1--1:14},
	publisher = {Schloss Dagstuhl - Leibniz-Zentrum f{\"{u}}r Informatik},
	series = {LIPIcs},
	title = {Constant-Factor {FPT} Approximation for Capacitated k-Median},
	url = {https://doi.org/10.4230/LIPIcs.ESA.2019.1},
	volume = {144},
	year = {2019}}

@inproceedings{badoiu-etal:approximate-clustering-coresets,
	author = {Mihai Bad\u{o}iu and Sariel Har{-}Peled and Piotr Indyk},
	bibsource = {dblp computer science bibliography, https://dblp.org},
	booktitle = {Proc. 34th Annual {ACM} Symposium on Theory of Computing (STOC'02)},
	doi = {10.1145/509907.509947},
	pages = {250--257},
	publisher = {{ACM}},
	title = {Approximate clustering via core-sets},
	url = {https://doi.org/10.1145/509907.509947},
	year = {2002}}

@inproceedings{cohen-addad-etal21:coreset-framework,
	author = {Vincent Cohen{-}Addad and David Saulpic and Chris Schwiegelshohn},
	bibsource = {dblp computer science bibliography, https://dblp.org},
	booktitle = {Proc. 53rd Annual {ACM} {SIGACT} Symposium on Theory of Computing (STOC'21)},
	doi = {10.1145/3406325.3451022},
	pages = {169--182},
	publisher = {{ACM}},
	title = {A new coreset framework for clustering},
	url = {https://doi.org/10.1145/3406325.3451022},
	year = {2021}}

@article{ding2020unified,
	author = {Ding, Hu and Xu, Jinhui},
	journal = {Algorithmica},
	number = {4},
	pages = {808--852},
	publisher = {Springer},
	title = {A unified framework for clustering constrained data without locality property},
	volume = {82},
	year = {2020}}

@article{ostrovsky2013effectiveness,
	author = {Ostrovsky, Rafail and Rabani, Yuval and Schulman, Leonard J and Swamy, Chaitanya},
	journal = {J.ACM},
	number = {6},
	pages = {1--22},
	publisher = {ACM New York, NY, USA},
	title = {The effectiveness of {Lloyd}-type methods for the $k$-means problem},
	volume = {59},
	year = {2013}}

@article{matouvsek2000approximate,
	author = {Matou{\v{s}}ek, Jir{\i}},
	journal = {Discrete \& Computational Geometry},
	number = {1},
	pages = {61--84},
	publisher = {Springer},
	title = {On approximate geometric $k$-clustering},
	volume = {24},
	year = {2000},
doi          = {10.1007/S004540010019}}

@inproceedings{awasthi2015hardness,
	author = {Awasthi, Pranjal and Charikar, Moses and Krishnaswamy, Ravishankar and Sinop, Ali Kemal},
	booktitle = {31st International Symposium on Computational Geometry (SoCG'15)},
	organization = {Schloss Dagstuhl-Leibniz-Zentrum fuer Informatik},
	title = {The Hardness of Approximation of Euclidean $k$-Means},
	year = {2015}}

@inproceedings{braverman2019coresets,
	author = {Braverman, Vladimir and Jiang, Shaofeng H-C and Krauthgamer, Robert and Wu, Xuan},
	booktitle = {International Conference on Machine Learning (ICML'19)},
	organization = {PMLR},
	pages = {744--753},
	title = {Coresets for ordered weighted clustering},
	year = {2019}}

@inproceedings{Cohen-AddadG0LL19,
	author = {Vincent Cohen{-}Addad and Anupam Gupta and Amit Kumar and Euiwoong Lee and Jason Li},
	bibsource = {dblp computer science bibliography, https://dblp.org},
	booktitle = {46th International Colloquium on Automata, Languages, and Programming, {ICALP} 2019, July 9-12, 2019, Patras, Greece},
	doi = {10.4230/LIPICS.ICALP.2019.42},
	pages = {42:1--42:14},
	publisher = {Schloss Dagstuhl - Leibniz-Zentrum f{\"{u}}r Informatik},
	series = {LIPIcs},
	title = {Tight {FPT} Approximations for k-Median and k-Means},
	url = {https://doi.org/10.4230/LIPIcs.ICALP.2019.42},
	volume = {132},
	year = {2019}}

@inproceedings{Jaiswal023,
  author       = {Ragesh Jaiswal and
                  Amit Kumar},

  title        = {Clustering What Matters in Constrained Settings: Improved Outlier
                  to Outlier-Free Reductions},
  booktitle    = {34th International Symposium on Algorithms and Computation, {ISAAC}
                  2023, December 3-6, 2023, Kyoto, Japan},
  series       = {LIPIcs},
  volume       = {283},
  pages        = {41:1--41:16},
  publisher    = {Schloss Dagstuhl - Leibniz-Zentrum f{\"{u}}r Informatik},
  year         = {2023},
  url          = {https://doi.org/10.4230/LIPIcs.ISAAC.2023.41},
  doi          = {10.4230/LIPICS.ISAAC.2023.41},
  bibsource    = {dblp computer science bibliography, https://dblp.org}
}

@article{BhattacharyaJK18,
  author       = {Anup Bhattacharya and
                  Ragesh Jaiswal and
                  Amit Kumar},
  title        = {Faster Algorithms for the Constrained k-means Problem},
  journal      = {Theory Comput. Syst.},
  volume       = {62},
  number       = {1},
  pages        = {93--115},
  year         = {2018},
  url          = {https://doi.org/10.1007/s00224-017-9820-7},
  doi          = {10.1007/S00224-017-9820-7},
  bibsource    = {dblp computer science bibliography, https://dblp.org}
}

@inproceedings{Huang0L025,
  author       = {Lingxiao Huang and
                  Jian Li and
                  Pinyan Lu and
                  Xuan Wu},
  title        = {Coresets for Constrained Clustering: General Assignment Constraints
                  and Improved Size Bounds},
  booktitle    = {Proceedings of the 2025 Annual {ACM-SIAM} Symposium on Discrete Algorithms,
                  {SODA} 2025, New Orleans, LA, USA, January 12-15, 2025},
  pages        = {4732--4782},
  publisher    = {{SIAM}},
  year         = {2025},
  url          = {https://doi.org/10.1137/1.9781611978322.161},
  doi          = {10.1137/1.9781611978322.161},
  bibsource    = {dblp computer science bibliography, https://dblp.org}
}

@inproceedings{CKLV17,
	author = {Flavio Chierichetti and Ravi Kumar and Silvio Lattanzi and Sergei Vassilvitskii},
	booktitle = {Advances in Neural Information Processing Systems 30: Annual Conference on Neural Information Processing Systems (NIPS'17)},
	pages = {5029--5037},
	title = {Fair Clustering Through Fairlets},
	year = {2017}}

@inproceedings{VC20,
	author = {Vincent Cohen{-}Addad},
	booktitle = {Proceedings of the 2020 {ACM-SIAM} Symposium on Discrete Algorithms (SODA'20)},
	pages = {2241--2259},
	publisher = {{SIAM}},
	title = {Approximation Schemes for Capacitated Clustering in Doubling Metrics},
	year = {2020}}

@inproceedings{ByrkaRU16,
	author = {Jaroslaw Byrka and Bartosz Rybicki and Sumedha Uniyal},
	bibsource = {dblp computer science bibliography, https://dblp.org},
	booktitle = {Integer Programming and Combinatorial Optimization - 18th International Conference, {IPCO} 2016, Li{\`{e}}ge, Belgium, June 1-3, 2016, Proceedings},
	doi = {10.1007/978-3-319-33461-5\_22},
	pages = {262--274},
	publisher = {Springer},
	series = {Lecture Notes in Computer Science},
	title = {An Approximation Algorithm for Uniform Capacitated k-Median Problem with 1+{\textbackslash}epsilon Capacity Violation},
	url = {https://doi.org/10.1007/978-3-319-33461-5\_22},
	volume = {9682},
	year = {2016}}

@inproceedings{ChuzhoyR05,
	author = {Julia Chuzhoy and Yuval Rabani},
	bibsource = {dblp computer science bibliography, https://dblp.org},
	booktitle = {Proceedings of the Sixteenth Annual {ACM-SIAM} Symposium on Discrete Algorithms, {SODA} 2005, Vancouver, British Columbia, Canada, January 23-25, 2005},
	pages = {952--958},
	publisher = {{SIAM}},
	title = {Approximating k-median with non-uniform capacities},
	url = {http://dl.acm.org/citation.cfm?id=1070432.1070569},
	year = {2005}}

@inproceedings{DemirciL16,
	author = {H. G{\"{o}}kalp Demirci and Shi Li},
	bibsource = {dblp computer science bibliography, https://dblp.org},
	booktitle = {43rd International Colloquium on Automata, Languages, and Programming, {ICALP} 2016, July 11-15, 2016, Rome, Italy},
	doi = {10.4230/LIPICS.ICALP.2016.73},
	pages = {73:1--73:14},
	publisher = {Schloss Dagstuhl - Leibniz-Zentrum f{\"{u}}r Informatik},
	series = {LIPIcs},
	title = {Constant Approximation for Capacitated k-Median with (1+epsilon)-Capacity Violation},
	url = {https://doi.org/10.4230/LIPIcs.ICALP.2016.73},
	volume = {55},
	year = {2016}}

@techreport{dasgupta08hardness-2-means,
	address = {San Diego, CA},
	author = {Dasgupta, Sanjoy},
	institution = {University of California, San Diego},
	number = {CS2008-0916},
	title = {The hardness of $k$-means clustering},
	type = {Technical Report},
	year = {2008}}

@inproceedings{BFRS15,
	author = {Jaroslaw Byrka and Krzysztof Fleszar and Bartosz Rybicki and Joachim Spoerhase},
	bibsource = {dblp computer science bibliography, https://dblp.org},
	booktitle = {Proceedings of the Twenty-Sixth Annual {ACM-SIAM} Symposium on Discrete Algorithms, {SODA} 2015, San Diego, CA, USA, January 4-6, 2015},
	doi = {10.1137/1.9781611973730.49},
	pages = {722--736},
	publisher = {{SIAM}},
	title = {Bi-Factor Approximation Algorithms for Hard Capacitated \emph{k}-Median Problems},
	url = {https://doi.org/10.1137/1.9781611973730.49},
	year = {2015}}

@inproceedings{DBLP:conf/stacs/Gadekar025,
  author       = {Ameet Gadekar and
                  Tanmay Inamdar},
  title        = {Dimension-Free Parameterized Approximation Schemes for Hybrid Clustering},
  booktitle    = {42nd International Symposium on Theoretical Aspects of Computer Science,
                  {STACS} 2025, March 4-7, 2025, Jena, Germany},
  series       = {LIPIcs},
  volume       = {327},
  pages        = {35:1--35:20},
  publisher    = {Schloss Dagstuhl - Leibniz-Zentrum f{\"{u}}r Informatik},
  year         = {2025},
  url          = {https://doi.org/10.4230/LIPIcs.STACS.2025.35},
  doi          = {10.4230/LIPICS.STACS.2025.35},
  bibsource    = {dblp computer science bibliography, https://dblp.org}
}

@article{CharikarGTS02,
  author       = {Moses Charikar and
                  Sudipto Guha and
                  {\'{E}}va Tardos and
                  David B. Shmoys},
  title        = {A Constant-Factor Approximation Algorithm for the k-Median Problem},
  journal      = {J. Comput. Syst. Sci.},
  volume       = {65},
  number       = {1},
  pages        = {129--149},
  year         = {2002},
  url          = {https://doi.org/10.1006/jcss.2002.1882},
  doi          = {10.1006/JCSS.2002.1882}
}

@article{kanungo2002efficient,
  title={An efficient k-means clustering algorithm: Analysis and implementation},
  author={Kanungo, Tapas and Mount, David M and Netanyahu, Nathan S and Piatko, Christine D and Silverman, Ruth and Wu, Angela Y},
  journal={IEEE transactions on pattern analysis and machine intelligence},
  volume={24},
  number={7},
  pages={881--892},
  year={2002},
  publisher={IEEE}
}

@article{BCPSS24,
  author       = {Nikhil Bansal and
                  Vincent Cohen{-}Addad and
                  Milind Prabhu and
                  David Saulpic and
                  Chris Schwiegelshohn},
  title        = {Sensitivity Sampling for k-Means: Worst Case and Stability Optimal
                  Coreset Bounds},
  journal      = {CoRR},
  volume       = {abs/2405.01339},
  year         = {2024},
  url          = {https://doi.org/10.48550/arXiv.2405.01339},
  doi          = {10.48550/ARXIV.2405.01339}
}

@inproceedings{DBLP:conf/stoc/Cohen-AddadEMN22,
  author       = {Vincent Cohen{-}Addad and
                  Hossein Esfandiari and
                  Vahab S. Mirrokni and
                  Shyam Narayanan},
  title        = {Improved approximations for Euclidean \emph{k}-means and \emph{k}-median,
                  via nested quasi-independent sets},
  booktitle    = {{STOC} '22: 54th Annual {ACM} {SIGACT} Symposium on Theory of Computing,
                  Rome, Italy, June 20 - 24, 2022},
  pages        = {1621--1628},
  publisher    = {{ACM}},
  year         = {2022},
  url          = {https://doi.org/10.1145/3519935.3520011},
  doi          = {10.1145/3519935.3520011}
}

@inproceedings{10.1145/TGOO,
author = {Thejaswi, Suhas and Gadekar, Ameet and Ordozgoiti, Bruno and Osadnik, Michal},
title = {Clustering with Fair-Center Representation: Parameterized Approximation Algorithms and Heuristics},
year = {2022},
isbn = {9781450393850},
publisher = {Association for Computing Machinery},
address = {New York, NY, USA},
url = {https://doi.org/10.1145/3534678.3539487},
doi = {10.1145/3534678.3539487},
booktitle = {Proceedings of the 28th ACM SIGKDD Conference on Knowledge Discovery and Data Mining},
pages = {1749–1759},
numpages = {11},
location = {Washington DC, USA},
series = {KDD '22}
}

@inproceedings{BravermanCJKST022,
  author       = {Vladimir Braverman and
                  Vincent Cohen{-}Addad and
                  Shaofeng H.{-}C. Jiang and
                  Robert Krauthgamer and
                  Chris Schwiegelshohn and
                  Mads Bech Toftrup and
                  Xuan Wu},
  title        = {The Power of Uniform Sampling for Coresets},
  booktitle    = {63rd {IEEE} Annual Symposium on Foundations of Computer Science, {FOCS}
                  2022, Denver, CO, USA, October 31 - November 3, 2022},
  pages        = {462--473},
  publisher    = {{IEEE}},
  year         = {2022},
  url          = {https://doi.org/10.1109/FOCS54457.2022.00051},
  doi          = {10.1109/FOCS54457.2022.00051},
  bibsource    = {dblp computer science bibliography, https://dblp.org}
}

@inproceedings{NarayananN19,
  author       = {Shyam Narayanan and
                  Jelani Nelson},
  title        = {Optimal terminal dimensionality reduction in Euclidean space},
  booktitle    = {Proceedings of the 51st Annual {ACM} {SIGACT} Symposium on Theory
                  of Computing, {STOC} 2019, Phoenix, AZ, USA, June 23-26, 2019},
  pages        = {1064--1069},
  publisher    = {{ACM}},
  year         = {2019},
  url          = {https://doi.org/10.1145/3313276.3316307},
  doi          = {10.1145/3313276.3316307}
}

@article{DabasGI25,
  author       = {Rajni Dabas and
                  Neelima Gupta and
                  Tanmay Inamdar},
  title        = {{FPT} approximation for capacitated clustering with outliers},
  journal      = {Theor. Comput. Sci.},
  volume       = {1027},
  pages        = {115026},
  year         = {2025},
  url          = {https://doi.org/10.1016/j.tcs.2024.115026},
  doi          = {10.1016/J.TCS.2024.115026},
  bibsource    = {dblp computer science bibliography, https://dblp.org}
}

@article{bandyapadhyay2024parameterized,
  title={Parameterized Approximation Algorithms and Lower Bounds for k-Center Clustering and Variants},
  author={Bandyapadhyay, Sayan and Friggstad, Zachary and Mousavi, Ramin},
  journal={Algorithmica},
  volume={86},
  number={8},
  pages={2557--2574},
  year={2024},
  publisher={Springer}
}

@inproceedings{ShmoysTA97,
  author       = {David B. Shmoys and
                  {\'{E}}va Tardos and
                  Karen I. Aardal},
  title        = {Approximation Algorithms for Facility Location Problems (Extended
                  Abstract)},
  booktitle    = {Proceedings of the Twenty-Ninth Annual {ACM} Symposium on the Theory
                  of Computing, 1997},
  pages        = {265--274},
  publisher    = {{ACM}},
  year         = {1997},
  url          = {https://doi.org/10.1145/258533.258600},
  doi          = {10.1145/258533.258600}
}

@article{CharikarG05,
  author       = {Moses Charikar and
                  Sudipto Guha},
  title        = {Improved Combinatorial Algorithms for Facility Location Problems},
  journal      = {{SIAM} J. Comput.},
  volume       = {34},
  number       = {4},
  pages        = {803--824},
  year         = {2005},
  url          = {https://doi.org/10.1137/S0097539701398594},
  doi          = {10.1137/S0097539701398594}
}

@article{KleinbergT02,
  author       = {Jon M. Kleinberg and
                  {\'{E}}va Tardos},
  title        = {Approximation algorithms for classification problems with pairwise
                  relationships: metric labeling and Markov random fields},
  journal      = {J. {ACM}},
  volume       = {49},
  number       = {5},
  pages        = {616--639},
  year         = {2002},
  url          = {https://doi.org/10.1145/585265.585268},
  doi          = {10.1145/585265.585268}
}

@inproceedings{BalcanBG09,
  author       = {Maria{-}Florina Balcan and
                  Avrim Blum and
                  Anupam Gupta},
  title        = {Approximate clustering without the approximation},
  booktitle    = {Proceedings of the Twentieth Annual {ACM-SIAM} Symposium on Discrete
                  Algorithms, {SODA}, 2009},
  pages        = {1068--1077},
  publisher    = {{SIAM}},
  year         = {2009},
  url          = {https://doi.org/10.1137/1.9781611973068.116},
  doi          = {10.1137/1.9781611973068.116}
}

\end{document}